\documentclass[10pt]{amsart}
\usepackage{amsmath}
  \usepackage{graphics} 
  \usepackage{epsfig} 
 \usepackage[colorlinks=true]{hyperref}
  \usepackage{float}
  \usepackage{caption}
\usepackage{subfig}

\newtheorem{theorem}{Theorem}[section]
\newtheorem{corollary}{Corollary}

\theoremstyle{definition}
\newtheorem{definition}[theorem]{Definition}
\newtheorem{remark}{Remark}

\title[ Predator interference]
      {Predator interference effects on biological control: The ``paradox" of the generalist predator revisited}

\author[  Bhowmick, Quansah, Basheer, Parshad, Upadhyay. ]{}


 \email{rparshad@clarkson.edu}

\begin{document}
\maketitle

\centerline{\scshape Suman Bhowmick, Emmanuel Quansah, Aladeen Basheer,
 and Rana D. Parshad}
{
\footnotesize
 \centerline{Department of Mathematics,}
 \centerline{Clarkson University,}
   \centerline{ Potsdam, New York 13699, USA.}
 }

   \medskip
\centerline{\scshape Ranjit Kumar Upadhyay }
{\footnotesize
 \centerline{Department of Applied Mathematics,}
 \centerline{Indian School of Mines,}
   \centerline{ Dhanbad 826004, Jharkhand, India.}

\medskip

}

\begin{abstract}
An interesting conundrum in biological control questions the efficiency of generalist predators as biological control agents.
Theory suggests, generalist predators are poor agents for biological control, primarily due to mutual interference. However
field evidence shows they are actually quite effective in regulating pest densities. 
In this work we provide a plausible answer to this paradox.
We analyze a three species model, where a generalist top predator is introduced into an ecosystem as a biological control, to check the population of a middle predator, that in turn is depredating on a prey species. We show that the inclusion of predator interference alone, can cause the solution of the top predator equation to blow-up in finite time, while there is global existence in the no interference case.
This result shows that interference could actually cause a population explosion of the top predator, enabling it to control the target species,  thus corroborating recent field evidence.
Our results might also partially explain the population explosion of certain species, introduced originally for biological control purposes, such as the cane toad (\emph{Bufo marinus}) in Australia, which now functions as a generalist top predator. 
We also show both Turing instability and spatio-temporal chaos in the model. Lastly we investigate time delay effects.
 \end{abstract}

 \section{Introduction}\label{1}
 \subsection{Background}
Exotic species, commonly referred to as ``invasive" species, are defined as any species, capable of propagating into a nonnative environment. If established,
they can be extremely difficult to eradicate, or even manage \cite{b19, M00}. There are numerous cases of environmental, economic
and ecological losses attributed to invasive species \cite{Pimentel05}. Some well-known examples of these species include the Burmese python in the United States and the cane toad in Australia, both of which have been wreaking havoc on indigenous ecosystems \cite{dorcas12, p18, Rodda05, letnic08}. Despite the magnitude of threats posed by invasive species, there have been very few conclusive results to eradicate or contain these species in real scenarios, in the wild. 

Biological control is an adopted strategy to limit invasive populations \cite{D96,V96}.
It works on the so called ``enemy release hypothesis", in which natural enemies of the target species/pest, are released against it in a controlled fashion. These can be in the form of predators, parasitoids, pathogens or combinations thereof \cite{V96}. It is a controversial yet fascinating area, with much debate and regulation. An interesting problem in the field, is the ``biological control paradox" \cite{A91}. This has attracted much research attention, see \cite{HK03} and the references within. The essential paradox here is that if one models a predator-pest system, via the Holling type II functional response, we cannot obtain a stable coexistence equilibrium, where the pest density is low. However, in reality many predators introduced for biological control purposes, are able to keep pest densities down to low levels.
Yet another interesting conundrum, also of a paradoxical nature, pertains to the effectiveness of generalist predators as biological controls.
In theory, generalist predators, are considered poor agents for biological control \cite{L97}. This is due to many factors, such as lack of specific prey targets, highly frequent interference amongst themselves, and their interference in the search of other specialist predators \cite{L97, P92}.
However, there is a large body of growing field evidence, showing that generalist predators, are actually quite effective in regulating pest densities \cite{R95, S99}. Thus there is an apparent discrepancy between what theory predicts, and what is actually seen in empirical observations. The primary objective of the current manuscript, is to corroborate these empirical observations, by proposing an alternate theory to understand the effect of generalist predator interferences, as it effects their efficiency as biological controls. 

Biological control is ``risky business". For example, the introduced species might attack a variety of species, other than those it was released to control. This phenomena is referred to as a \emph{non-target effect} \cite{F00}, and is common in natural enemies that are generalist predators. Some well known example of this are the cane toad (\emph{Bufo marinus}) and the Nile Perch (\emph{Lates niloticus}) \cite{CaneToad, NP05}. 
The secondary objective of the current manuscript, is to use our proposed theory to explain why certain species such as the cane toad, that were originally introduced as a biological control, have had an explosive increase in population. Before we delve further into these aspects we briefly survey some of the relevant literature on mutual interference.

Mutual interference \cite{1, 2a} is defined as the behavioral interactions among feeding organisms, that reduce the time that each individual spends obtaining food, or the amount of food each individual consumes. It occurs most commonly where the amount of food is scarce, or when the population of feeding organisms is large \cite{3a}.
Food chain models incorporating mutual interference have a long history, and were first proposed by Hassell \cite{1}, and Roger and Hassell \cite{2}, to model insect parasites and predator searching behaviour. Three species food chain models with mutual interference, and time delays, were proposed by Freedman and his group and they studied the trade-off between mutual interference among predators, and time delays due to gestation \cite{3, 4}. They observed that mutual interference is acting as a stabilizing factor and time delay does not necessarily destabilize the system, but increasing delay may cause a bifurcation into periodic solutions. For a delayed predator-prey model with mutual interference parameter Wang and Zu \cite{5, 6} have obtained some sufficient conditions for the permanence and global attractively. Comparing this with  empirical/statistical evidences from 19 predator-prey systems, Skalski and Gilliam \cite{7} pointed out that the predator dependent functional responses (Hassell-Varley type \cite{8}, Beddington-DeAngelis  type  \cite{9,10,11,12,13,14} and Crowley-Martin type \cite{15}) could provide a better description of predator feeding, over a range of predator-prey abundances, and in some cases the Beddington-DeAngelis (BD) type functional response performed even better \cite{16}. Upadhyay and Iyengar \cite{17} have pointed out that if predators do not waste time, interacting with one another, or if the attacks are always successful and instantaneous, then the response changes into Holling type II functional response, and the predators benefit from co-feeding. 

In light of our objectives we ask the following questions:

 \begin{itemize}
\item How does mutual interference amongst generalist predators effect their efficiency as biological controls?
\item Specifically, can the interference delay or exacerbate controlling the target species?
\item Can interference cause a population explosion in certain generalist predator biological agent populations?
\end{itemize}


 To answer these we investigate the three species model proposed in \cite{20}, where the interaction between the intermediate predator and the top predator is modeled according to the Beddington-DeAngelis type functional response. This response models predator interference in the top predator $r$ \cite{9,10}.
We first introduce a concept central to our investigations.

\subsection{An alternate concept of a ``successful"  biological control}

We introduce the dynamic of finite time blow-up to address our question on interference, via the following connected definitions:

\begin{definition}[finite time blow-up]
 Given a mathematical model for a nonlinear process, say through a partial differential equation (PDE), one says finite time blow-up occurs if
\begin{equation*}
\lim_{t\rightarrow T^{\ast}<\infty}\| r \|_{X} \rightarrow \infty,
\end{equation*}
where $X$ is a certain function space with a norm $\|\cdot\|$, $r$ is the solution to the PDE in question, and $T^{\ast}$ is the blow-up time.
In the case of an ordinary differential equation (ODE) model the function space $X$ is simply the real numbers $\mathbb{R}$. If blow-up does not occur, that is $T^{*}=\infty$,
we say there is global existence, that is a solution exists for all time.
\end{definition}

\begin{definition}[excessive population of the biological control]
Consider a mathematical model (PDE or ODE) for the population dynamics of a certain species $r$, introduced as a biological control.
 If the model blows-up in finite time, that is
\begin{equation*}
\lim_{t\rightarrow T^{\ast}<\infty}\| r \|_{X} \rightarrow \infty.
\end{equation*}
then we say the $r$ population has reached an ``excessive'' level. In these excessive numbers it is able to wipe out the target pest almost with certainty. 
\end{definition}

There is a rich history of blow-up problems in PDE theory and its interpretations in physical phenomenon. For example, blow-up may be interpreted as the failure of certain constitutive materials leading to gradient catastrophe or fracture, it may be interpreted as an uncontrolled feedback loop such as in thermal runaway, leading to explosion. It might also be interpreted as a sudden change in physical quantities such as pressure or temperature such as during a shock or in the ignition process. The interested reader is referred to \cite{QS7, S98}. Blow-up in population dynamics is usually interpreted as excessively high concentrations in small regions of space, such as seen in chemotaxis problems \cite{H09}. In the current manuscript, blow-up in a population of a population of bio-control agents, is interpreted as the population becoming excessive, which then enables it to easily control the target pest, given that it has excessive numbers.
This leads us to equate

\begin{equation*}
\mbox{blow-up in finite time}  \ = \  \mbox{control agent being successful}. 
\end{equation*}

However, there might also be various negative effects associated with this \cite{PQB15, PKK15}.
For example The following chain of events might occur

\begin{eqnarray*}
&& =   \mbox{control agent is successful in eradicating target pest} \\
&& = \ \mbox{However the bio-control agent now has a ``very large" population} \\
&& =  \ \mbox{These agents being generalists start to attack other native species} \\
&& =  \ \mbox{This wreaks havoc on the ecosystem. That is now there are non-target effects}. \\
&& = \ \mbox{The bio-control agent population is now ``uncontrollable"} \\
\end{eqnarray*}

Via this formulation, it is conceivable to see how this is what might have occurred with the cane toad in Australia. The cane toad was introduced in Australia in 1935 to control the cane beetle. However, the toad being a generalist predator, attacked various other species \cite{p18}. In addition, the toad is highly poisonous and therefore other predators shy away from eating it. This has enabled the toad population to grow virtually unchecked, and it is today considered one of Australia's worst invasive species \cite{CaneToad}.

The contributions of the current manuscript are: 

\begin{itemize}

\item We introduce a new concept to measure the success of a biological control by equating
\begin{equation*}
\mbox{blow-up in finite time}  \ = \  \mbox{control agent being successful}. 
\end{equation*}

\item We show the three species model proposed in \cite{20}, and its spatially explicit version, can blow-up in finite time, for sufficiently large initial data via theorem \ref{t1} and corollary \ref{ct1}.

\item We show predator interference \emph{is the sole factor} in inducing blow-up, when there is global existence in the no interference case. This is demonstrated via theorem \ref{t2} and corollary \ref{ct2}.

\item  The spatially explicit form of the three species model proposed in \cite{20}, possesses spatio-temporal chaos. Also time delays in the temporal model, affect both the blow-up dynamics as well as the chaotic dynamics.

\item Based on the above results, we propose that generalist predator interference \emph{might be a cause of their success}, in controlling target pests. However, predator interference may also be a cause of the population explosion of certain species, introduced originally for biological control purposes, such as the cane toad in Australia. 
\end{itemize}

We provide details of the model formulation next. 

\section{Model formulation}
\label{2}

We consider an ecosystem where a specialist predator/pest $v$, invasive or otherwise, is depredating on a prey species $u$. In order to control $v$, a generalist predator $r$ is released into the same ecosystem. The goal is that $r$ will hunt and depredate on $v$, its favorite food, thus lowering the population of $v$.
The dynamical interaction between $v$ and $u$ is modeled via standard Rosenzweig-McArthur scheme \cite{R73}, while the dynamics between $r$ and $v$ is modeled via a Leslie-Gower formulation, where interference in the top predator $r$ is assumed, and this is modeled via a Beddington DeAngelis functional response \cite{10}. 
Upadhyay et al. \cite{20} have proposed a tri-trophic hybrid food chain model to study such a system. Here the prey population density $u$ serves as the only food for the intermediate specialist predator population of density $v$. This population $v$ serves as a favorite food for the generalist top predator population of density r.

\begin{eqnarray}
\label{eq:1}
\frac{du}{dt}&=&u(1-u)-\frac{uv}{u+\omega_4}\\
\frac{dv}{dt}&=&-\omega_5 v+\frac{\omega_6 uv}{u+\omega_7}-\frac{vr}{v+\omega_8r+\omega_9}\\
\frac{dr}{dt}&=&\omega_{10}r^2-\frac{\omega_{11}r^2}{v+\omega_{12}}
\end{eqnarray}

Now, we present a brief sketch of the construction of the model which is biologically motivated.
\begin{itemize}
	\item	Behavior of the entire community is assumed to arise from the coupling of these interacting species where population $v$ prey on $u$ and only on $u$ and the top predator preys on favorite food $v$ but it has other options for food in case of scarcity or short supply.   
	\item	The rate of change of population size for prey and intermediate predator has been written according to the Rosenzweig-McArthur scheme i.e. predator population dies out exponentially in the absence of its prey and prey population density grows according to the famous logistic growth rate.
	\item	The top predator is a sexually reproducing species. The interaction between the intermediate predator $v$ and the top predator $r$ is according to the Beddingto-DeAngelis type functional response \cite{10}. This response models predator interference in the top predator $r$. 
\end{itemize}

We impose different assumptions from the ones assumed in \cite{20} to formulate the spatially explicit form of the differential equations which describe the model. The detailed meaning of the different parameters is given in \cite{20}. Further, we also assume that all the three populations perform active movements in space. Random movement of animals occurs because of various requirements and necessities like, search for better food, better opportunity for social interactions such as finding mates \cite{21}. Food availability and living conditions demand that these animals migrate to other spatial locations. In the proposed model, we have included diffusion terms assuming that the animal movements are uniformly distributed in all directions. 

The model is described by the following set of partial differential equations.  
with suitable initial conditions and Neumann boundary conditions,

\begin{eqnarray}
\label{eq:1dc}
\frac{\partial u}{\partial t}&=& d_{1}\Delta u + u(1-u)-\frac{uv}{u+\omega_4}\\ \nonumber
\frac{\partial v}{\partial t}&=& d_{2}\Delta v -\omega_5 v+\frac{\omega_6 uv}{u+\omega_7}-\frac{vr}{v+\omega_8r+\omega_9}\\ \nonumber
\frac{\partial r}{\partial t}&=&d_{3}\Delta r + \omega_{10}r^2-\frac{\omega_{11}r^2}{v+\omega_{12}}  \\ \nonumber
\end{eqnarray}
                                     
The physical domain for the problem is some bounded set $\Omega$, which is a subset of $\mathbb{R}^{n}$. For numerical simulations, we restrict ourselves to $n=1$, here we take $\Omega =[o,\pi]$.                                                        
 $d_{1}, d_{2}, d_{3}$represent the diffusion coefficients of prey, intermediate predator and top predator respectively. $\Delta$ is the Laplacian operator which describes the random movement.  
 Note we always restrict $\omega_{10} < \frac{\omega_{11} }{\omega_{12}}$. Violation of this condition, will \emph{always} cause finite time blow-up \cite{PK13}, and as a rule must be avoided in such models \cite{U97}.
  $\textbf{n}$ is the outward unit normal vector of the boundary . The main reason for choosing such boundary conditions is that we are interested in the self-organization of patterns, and zero-flux boundary conditions imply that no external input is imposed \cite{22}. 
 
   \begin{remark}
 A pest only free equilibrium $(r^{*},0,u^{*})$ does not exist in the ODE or PDE case. This is due to the assumption that $\omega_{10} < \frac{\omega_{11} }{\omega_{12}}$. In the event that $v=0$, the equation for the top predator, has only a decay term $\left( \omega_{10} - \frac{\omega_{11} }{\omega_{12}} \right) r^2 < 0$, hence $r$ will also be driven to zero. However, in the event that $r$ blows-up (gets excessively large), the $v$ equation will reduce to
 
 \begin{equation}
\label{eq:1dcn}
\frac{\partial v}{\partial t}= d_{2}\Delta v -\omega_5 v+\frac{\omega_6 uv}{u+\omega_7}-\frac{1}{\omega_8}v
\end{equation}

In this case $v$ can be driven to extinction if $\omega_5 + \frac{1}{\omega_8} > \omega_{6}$
  
 \end{remark}

  \begin{remark}
 Note, one could also consider the addition of $K\omega_8\Delta (r^2)= \frac{\omega_8}{2} \nabla(r\nabla r)$, in the equation for $r$. This would represents additional dispersive force due to local overcrowding. This is interpreted as movement from high towards low concentrations of $r$, directly proportional to $r$. Hence, $r$ attempts to disperses toward lower concentrations, due to the excessive intraspecific competition from interference. 
 There is a fair amount of literature that in high prey density areas intraspecific competition in predator species is heightened \cite{emsens14}.  In such a setting the predators should have greater dispersal in order to better assimilate available resources and avoid crowding. Furthermore $K$ can be set to 0 so that we recover $\eqref{eq:1dc}$.
\end{remark}

%
%

\section{Finite time blow-up}

In this section we show improvements to Theorem \ref{thm:u} from \cite{20}. We first recap the result from \cite{20}

\begin{theorem}
\label{thm:u}
Assume that the following conditions hold

\begin{equation}
\omega_{7} \geq \omega_{4}, \frac{\omega_{10}}{\omega_{5}} < \omega_{5} \omega_{11}\delta
\end{equation}

and let $A$ be defined by the set

\begin{equation}
A = [(u,v,r) \in \mathbb{R}^{3}_{+}: 0 \leq u \leq 1,  0 \leq u + \frac{v}{\omega_{6}} \leq v_{c}, 0 \leq u + \frac{v}{\omega_{6}} +\delta r \leq r_{c}]
\end{equation}

where $\delta=\frac{4}{\omega_{5}[4\omega_{5}(\omega_{6}+\omega_{12})+\omega_{6}]}$ and $M= \frac{(\omega_{5})^{2}\delta}{4(\omega_{5})^{2}\omega_{11}\delta - \omega_{10}}$ 
\newline
and $v_{c}=1+\frac{1}{4\omega_{5}}$  and $r_{c}=1+\frac{1}{4\omega_{5}}+\frac{M}{\omega_{5}}$ then
\newline
(i) $A$ is positively invariant
\newline
(ii) All non negative solutions of system \eqref{eq:1} are uniformly bounded and they eventually enter the attracting set $A$
\newline
(iii) The model system \eqref{eq:1} is dissipative

\end{theorem}

\begin{remark}
The primary issue with Theorem \ref{thm:u} is that an attracting set is invariant, but an invariant set \textbf{may not} be attracting. Although the set $A$ is invariant, that is if we start in $A$ we remain in $A$ for all time, we will show it is not attracting for large initial conditions. In particular for large enough initial conditions, system \eqref{eq:1} can actually blow-up in finite time.
\end{remark}

Next we demonstrate finite time blow-up in \eqref{eq:1}.

\begin{theorem}
\label{t1}
Consider the three species food chain model given by \eqref{eq:1}, for any choice of parameters, including the ones satisfying Theorem  \ref{thm:u}, and a $\delta_{1} > 0$, such that $\omega_{10} > \delta_{1}$.  Given any initial data $u_{0}$, there exists initial data $(v_{0},r_{0})$, such that if this data meets the largeness condition

\begin{equation}
\frac{\omega_{5}+\frac{1}{\omega_{8}}}{\delta} < |r_{0}| \ln\left( \frac{|v_{0}|}{\frac{\omega_{11}}{\omega_{10} - \delta_{1}} - \omega_{12}}\right)
\end{equation}

then \eqref{eq:1} will blow-up in finite time, that is

\begin{equation*}
\lim_{t\rightarrow T^{\ast}<\infty}\| r \| \rightarrow \infty.
 \end{equation*}

Here the blow-up time $T^{*} \leq \frac{1}{\delta_{1}|r_{0}|}$

\end{theorem}

\begin{proof}
Consider the equation for the top predator

\begin{equation*}
\frac{dr}{dt}=\left(\omega_{10}-\frac{\omega_{11}}{v+\omega_{12}}\right)r^{2}
\end{equation*}

In the event that $\omega_{10} > \frac{\omega_{11}}{\omega_{12}}$, blow-up is obvious. If $\omega_{10} < C \frac{\omega_{11}}{\omega_{12}}$, where $C < 1$ or $C << 1$, blow-up is far from obvious. However still possible for large data. To see this note, if

\begin{equation*}
\left(\omega_{10}-\frac{\omega_{11}}{v+\omega_{12}}\right) > \delta_{1} > 0,
\end{equation*}

then $r$ will blow-up in finite time in comparison with

\begin{equation*}
\frac{dr}{dt}=\delta_{1} r^{2}
\end{equation*}

The tricky part here is that $\left(\omega_{10}-\frac{\omega_{11}}{v+\omega_{12}}\right)$ can switch sign, and this is dependent on the dynamics of the middle predator $v$, which changes in time.
In order to guarantee blow-up, we must have that $\left(\omega_{10}-\frac{\omega_{11}}{v+\omega_{12}}\right) > \delta_{1} $ or equivalently we must guarantee that

\begin{equation}
\label{eq:ve}
v > \frac{\omega_{11}}{\omega_{10} - \delta_{1}} - \omega_{12}.
\end{equation}

To this end we will work with the equation for $v$

\begin{equation*}
\frac{dv}{dt} = -\omega_5 v+\frac{\omega_6 uv}{u+\omega_7}-\frac{vr}{v+\omega_8r+\omega_9},
\end{equation*}

here the interference term $-\frac{vr}{v+\omega_8r+\omega_9}$ will come to our aid. Note

\begin{equation*}
-\frac{vr}{v+\omega_8r+\omega_9} >  - \frac{v}{\omega_8}
\end{equation*}

as trivially one has

\begin{equation*}
\frac{r}{v+\omega_8r+\omega_9} < \frac{1}{\omega_{8}}.
\end{equation*}

Thus using the above we obtain

\begin{equation*}
\frac{dv}{dt} > -\omega_5 v-\frac{v}{\omega_8}
\end{equation*}

Multiplying the above through by $e^{\left(\omega_5 + \frac{1}{\omega_8}\right)t}$, and integrating in the time interval $[0,t]$ yields the estimate

\begin{equation*}
|v| > e^{-\left(\omega_5 +\frac{1}{\omega_8}\right)t}|v_{0}|
\end{equation*}

we now use this in \eqref{eq:ve} to yield

\begin{equation*}
|v| > e^{-\left(\omega_5 +\frac{1}{\omega_8}\right)t}|v_{0}| > \frac{\omega_{11}}{\omega_{10} - \delta_{1}} - \omega_{12}.
\end{equation*}

Equivalently, 

\begin{equation*}
 \ln\left( \frac{|v_{0}|}{\frac{\omega_{11}}{\omega_{10} - \delta_{1}} - \omega_{12}} \right) > t\left(\omega_5+\frac{1}{\omega_8}\right).
\end{equation*}

 Now note

\begin{equation*}
\frac{dr}{dt}=\delta_{1} r^{2}
\end{equation*}

blows-up at time $T^{*} = \frac{1}{\delta_{1}|r_{0}|}$, thus if we choose data such that

\begin{equation*}
 \ln\left( \frac{|v_{0}|}{\frac{\omega_{11}}{\omega_{10} - \delta_{1}} - \omega_{12}} \right) \frac{1}{\left(\omega_5 +\frac{1}{\omega_8}\right)} > t > T^{*} = \frac{1}{\delta_{1}|r_{0}|},
\end{equation*}

Then the above guarantees that $v$ will remain above the critical level $\frac{\omega_{11}}{\omega_{10} - \delta_{1}} - \omega_{12}$, for sufficiently long enough time, for $r$ to blow-up. This yields that as long as the following holds

\begin{equation*}
|r_{0}| \ln\left( \frac{|v_{0}|}{\frac{\omega_{11}}{\omega_{10} - \delta_{1}} - \omega_{12}} \right) >\frac{\left(\omega_5 +\frac{1}{\omega_8}\right)}{\delta_{1}},
\end{equation*}

$r$ will blow up in finite time, independent of the choice of the parameters of theorem \ref{thm:u}. This proves the Theorem.

\end{proof}

\begin{remark}
Note that here we have an \emph{explicit} sufficient condition on the largeness of the initial data to cause blow-up. This is stronger than our results on blow-up for similar models \cite{PKK15}, where we could only prove blow-up for data large enough, but not quantify how large the data was required to be.
\end{remark}

\begin{corollary}
\label{ct1}
Consider the diffusive  three species food chain model given by \eqref{eq:1dc}, for any choice of parameters, including the ones satisfying theorem \ref{thm:u}, and a $\delta_{1} > 0$.  Given initial data $u_{0}$, there exists initial data $(v_{0},r_{0})$, such that if this data meets the largeness condition

\begin{equation}
\frac{\omega_{5}+\frac{1}{\omega_{8}}}{\delta} <||r_{0}||_{\infty} \ln\left( \frac{||v_{0}||_{\infty}}{\frac{\omega_{11}}{\omega_{10} - \delta_{1}} - \omega_{12}}\right)
\end{equation}

then \eqref{eq:1dc} will blow-up in finite time, that is

\begin{equation*}
\lim_{t \rightarrow T^{\ast} < \infty}||r||_{\infty} \rightarrow \infty.
 \end{equation*}

Here the blow-up time $T^{*} \leq  \frac{1}{\delta_{1}||r_{0}||_{\infty}}$

\end{corollary}

\begin{proof}
The proof in this PDE case follows via a simple comparison argument \cite{S83}, used in conjunction with theorem \ref{t1}.
\end{proof}

\begin{remark}
The presence of blow-up is not affected via the overcrowding term, if we maintain Neumann boundary conditions at the boundaries.  This can be seen in a straightforward manner.  Let us assume that the model with no overcrowding blows up.  
Now, consider the integration over $\Omega$ of the equation for $r$, \emph{if one assumes overcrowding} in \eqref{eq:1dc}, that is an addition of a  $\Delta(r^2)$ term to the equation for the top predator.
 Since the overcrowding term integrates to zero due to the boundary condition, that is

\begin{equation}
\label{eq:H12}
\int_{\Omega} \Delta(r^2) dx = \int_{\partial \Omega}r\nabla r \cdot \textbf{n}ds =  \int_{\partial \Omega}r ( 0 ) ds = 0,
\end{equation}
This implies that if we set $H(t)= \int_{\Omega}r(x,t)dx$, then,

\begin{equation}
\label{eq:H}
\frac{d}{dt}H(t) \geq \frac{\delta}{\sqrt{|\Omega|}} H(t)^2,
\end{equation}
as there exists a $\delta$ such that $\left(\omega_{10}-\frac{\omega_{11}}{v+\omega_{12}} \right)> \delta > 0$. 
This leads to the blow-up of $H(t)$. 
\end{remark}

We next show rigorously the effect of interference on the blow-up, via the interference parameter $\omega_{8}$. This is stated via the following theorem.

\begin{theorem}
\label{t2}
Consider the three species food chain model given by \eqref{eq:1}, for any choice of parameters, including the ones satisfying theorem \ref{thm:u}. Given initial data $(u_{0},v_{0},r_{0})$, such that there exists a globally existing solution to the no interference case, that is when $\omega_8=0$, the food chain model given by \eqref{eq:1}, with $\omega_8>0$,  will blow up in finite time, that is

\begin{equation*}
\lim_{t\rightarrow T^{\ast}<\infty}| r | \rightarrow \infty.
 \end{equation*}

as long as there is a positive time $T_{1}$, s.t. the interference threshold $\omega_{8}>0$, satisfies the following 

\begin{equation}
e^{\left(\omega_5 +\frac{1}{\omega_8}\right)T_{1}}
= |v_{0}|\left(\omega_5 +\frac{1}{\omega_8}\right)  \left(   \left(\frac{\omega_{10}}{\omega_{11}}\right)T_{1} -\frac{1}{(\omega_{11}r_{0})} \right)+ 1.
 \end{equation}
 \end{theorem}

\begin{proof}
Without interference that is when $\omega_{8}=0$, the model reduces to

\begin{eqnarray}
\frac{du}{dt}&=&u(1-u)-\frac{uv}{u+\omega_4}\\
\frac{dv}{dt}&=&-\omega_5 v+\frac{\omega_6 uv}{u+\omega_7}-\frac{vr}{v+\omega_9}\\
\frac{dr}{dt}&=&\omega_{10}r^2-\frac{\omega_{11}r^2}{v+\omega_{12}}.
\end{eqnarray}

Consider \eqref{eq:1}, with positive initial conditions $(u_{0},v_{0},r_{0})$.

By integrating the equation for the top predator we obtain

\begin{equation}
\label{eq:r1}
r=\frac{1}{\frac{1}{r_{0}}-\omega_{10}t +\int^{t}_{0}\frac{\omega_{11}}{v(s)+\omega_{12}}ds}.
\end{equation}

Since we have a globally existing solution, we must have that
the continuous function:
\begin{equation}
 \psi \left( t\right) =\frac{1}{r_{0}}-\omega_{10}t +\int^{t}_{0}\frac{\omega_{11}}{v(s)+\omega_{12}}ds > 0
 \end{equation}

else, $r$ would blow-up in finite time.
Now if there is interference, the system is given by \eqref{eq:1}, and here we have the same lower estimate on $v$ via \eqref{eq:ve}

\begin{equation*}
|v| > e^{-\left(\omega_5 -\frac{1}{\omega_8}\right)t}|v_{0}|.
\end{equation*}

We will now show that under this dynamics, it is always possible to choose a $\omega_8$ (if certain additional restrictions are met) s.t the continuous functional

\begin{equation}
 \psi \left( t\right) =\frac{1}{r_{0}}-\omega_{10}t +\int^{t}_{0}\frac{\omega_{11}}{v(s)+\omega_{12}}ds = 0, \ \mbox{for} \ \mbox{some} \ t=T_{1},
 \end{equation}

plugging this into \eqref{eq:r1} then, immediately implies the finite time blow-up of $r$.

In order for $\psi (t) =0$, we require that

\begin{equation}
\label{eq:int1}
\int^{t}_{0}\frac{\omega_{11}}{v(s)+\omega_{12}}ds = \omega_{10}t -\frac{1}{r_{0}}
 \end{equation}

using \eqref{eq:ve} we obtain

\begin{equation*}
v +  \omega_{12} > e^{-\left(\omega_5 +\frac{1}{\omega_8}\right)t}|v_{0}| + \omega_{12},
\end{equation*}

thus

\begin{equation}
\int^{t}_{0}\frac{\omega_{11}}{v(s)+\omega_{12}}ds <  \int^{t}_{0}\frac{\omega_{11} e^{\left(\omega_5 +\frac{1}{\omega_8}\right)s}  }{|v_{0}| + \omega_{12}  e^{\left(\omega_5 +\frac{1}{\omega_8}\right)s}} ds= \omega_{10}t -\frac{1}{r_{0}}.
 \end{equation}
 
 Now we have that
 
 \begin{equation}
 \int^{t}_{0}\frac{\omega_{11} e^{\left(\omega_5 +\frac{1}{\omega_8}\right)s}  }{|v_{0}| + \omega_{12}  e^{\left(\omega_5 +\frac{1}{\omega_8}\right)s}} ds <
  \frac{\omega_{11}}{|v_{0}|}\int^{t}_{0} e^{\left(\omega_5 +\frac{1}{\omega_8}\right)s}ds 
  = \frac{\omega_{11}}{|v_{0}|\left(\omega_5 +\frac{1}{\omega_8}\right)} \left( e^{\left(\omega_5 +\frac{1}{\omega_8}\right)t} -1   \right).
 \end{equation}

Hence we will have blow-up if there exists an $\omega_{8}$ s.t at some $t=T_{1}$ we have
\begin{equation}
\frac{\omega_{11}}{|v_{0}|\left(\omega_5 +\frac{1}{\omega_8}\right)} \left( e^{\left(\omega_5 +\frac{1}{\omega_8}\right)T_{1}} -1   \right) = \omega_{10}T_{1} -\frac{1}{r_{0}}
 \end{equation}
 
 or equivalently
 
 \begin{equation}
e^{\left(\omega_5 +\frac{1}{\omega_8}\right)T_{1}}
= |v_{0}|\left(\omega_5 +\frac{1}{\omega_8}\right)  \left(   \left(\frac{\omega_{10}}{\omega_{11}}\right)T_{1} -\frac{1}{(\omega_{11}r_{0})} \right)+ 1,
 \end{equation}

as this will imply

\begin{equation}
\int^{T_{1}}_{0}\frac{\omega_{11}}{v(s)+\omega_{12}}ds < \frac{\omega_{11}}{|v_{0}|\left(\omega_5 +\frac{1}{\omega_8}\right)} \left( e^{\left(\omega_5 +\frac{1}{\omega_8}\right)T_{1}} -1   \right) =  \omega_{10}T_{1} -\frac{1}{r_{0}},
 \end{equation}

which would mean $ \psi \left( T_{1}\right) < 0$. Note $\psi(0)>0$ and thus by continuity of $ \psi$, and a direct application of the intermediate value theorem, there must exist a time $t=T_{2}$, $T_{2} < T_{1}$, s.t $ \psi \left( T_{2}\right) = 0$, and thus subsequently $r=\frac{1}{\psi(t)}$ must blow-up at $t=T_{2}$. This proves the theorem.
\end{proof}

%

\begin{corollary}
\label{ct2}
Consider the diffusive three species food chain model given by \eqref{eq:1dc}, for any choice of parameters, including the ones satisfying Theorem \ref{t1}. Given initial data $(u_{0},v_{0},r_{0})$, such that there exists a globally existing solution to the no interference case, that is when $\omega_8=0$, the food chain model given by \eqref{eq:1dc} will blow up in finite time, that is

\begin{equation*}
\lim_{t\rightarrow T^{\ast}<\infty}|| r ||_{\infty} \rightarrow \infty.
 \end{equation*}

as long as there exists an interference threshold $\omega_{8}>0$, such that there is a positive time $T_{1}$ that solves the following 

\begin{equation}
e^{\left(\omega_5 +\frac{1}{\omega_8}\right)T_{1}}
= |v_{0}|\left(\omega_5 +\frac{1}{\omega_8}\right)  \left(   \left(\frac{\omega_{10}}{\omega_{11}}\right)T_{1} -\frac{1}{(\omega_{11}r_{0})} \right)+ 1.
 \end{equation}

\end{corollary}

\begin{proof}
The proof in this PDE case follows via a simple comparison argument \cite{S83}, used in conjunction with theorem \ref{t2}.
\end{proof}

\begin{remark}
Again, the above result is true irrespective of overcrowding in \eqref{eq:1dc}. 
 Since the overcrowding term integrates to zero due to the boundary condition.
 \end{remark}

\section{Turing Instability}
In this section we investigate Turing instability in \eqref{eq:1dc}. We uncover both spatial and spatio-temporal patterns, and provide the details of the Turing analysis. 
We derive conditions where the unique positive interior equilibrium point $(u^*,v^*,r^*)$ is stable in the absence of diffusion, and unstable due to the action of diffusion, with a small perturbation to the positive interior equilibrium point. We first linearize model \eqref{eq:1dc} about the homogeneous steady state, we introduce both space and time-dependent fluctuations around $(u^*,v^*,r^*)$. This is given as
\begin{subequations}\label{eq:7}
\begin{align}
u=u^* +  \hat{u}(\xi,t),\\
v=v^* + \hat{v}(\xi,t),\\
r=r^*  +  \hat{r}(\xi,t),
\end{align}
\end{subequations}
where $| \hat{u}(\xi,t)|\ll u^*$, $| \hat{v}(\xi,t)|\ll v^*$ and  $| \hat{r}(\xi,t)|\ll r^*$. Conventionally we choose
\[
\left[ {\begin{array}{cc}
\hat{u}(\xi,t)  \\
\hat{v}(\xi,t) \\
\hat{r}(\xi,t)
\end{array} } \right]
=
\left[ {\begin{array}{cc}
\epsilon_1  \\
\epsilon_2 \\
\epsilon_3
\end{array} } \right]
e^{\lambda t + ik\xi},
\]
where  $\epsilon_i$ for $i=1,2,3$ are the corresponding amplitudes, $k$ is the wavenumber, $\lambda$ is the growth rate of perturbation in time $t$ and $\xi$ is the spatial coordinate.
Substituting \eqref{eq:7} into \eqref{eq:1dc} and ignoring higher order terms including nonlinear terms, we obtain the characteristic equation  which  is given as
\begin{align}\label{eq:1.2.10}
({\bf J} - \lambda{\bf I} - k^2{\bf D})
\left[ {\begin{array}{cc}
\epsilon_1  \\
\epsilon_2 \\
\epsilon_3
\end{array} } \right]=0,
\end{align}
where
\[
\quad
\bf {D} =
\left[ {\begin{array}{ccc}
d_1 & 0     & 0  \\
0     & d_2 & 0 \\
0     & 0     & d_3
\end{array} } \right],
\]
$ \bf{J}= \begin{bmatrix}
     u^* \left(-1+ \frac{v^*}{\alpha_{*}^{2}}\right) &   \frac{u^*}{\alpha_{*}^{2}} &  0 \\
         \frac{\omega_6\omega_7v^*}{\beta_{*}^{2}}& \frac{v^*r^*}{\gamma_{*}^{2}}   & -\frac{v^{*}(v^*+\omega_9)}{\gamma_{*}^{2}} \\
            0& \frac{(\omega_{10}r^*)^2}{\omega_{11}} & 0\\
     \end{bmatrix}
      =\begin{bmatrix}
       J_{11} & J_{12} & J_{13}\\
       J_{21} & J_{22} & J_{23}\\
        J_{31} & J_{32} & J_{33}\\
       \end{bmatrix},
$    \\
\\and $\bf{I}$ is a $3\times 3$ identity matrix.\\
For the non-trivial solution of \eqref{eq:1.2.10}, we require that
\[
\left|
\begin{array}{ccc}
J_{11}-\lambda -k^2d_1 & J_{12}                            & J_{13}\\
       J_{21}                     & J_{22}-\lambda -k^2d_2 & J_{23}\\
       J_{31}                     & J_{32}                             & J_{33}-\lambda -k^2d_3\\
 \end{array} \right|=0,
\]
which gives a dispersion relation corresponding to $(u^*,v^*,r^*)$ where $\alpha_{*}=u^*+\omega_4$, $\beta_{*}=u^*+\omega_{7}$,$\gamma_{*}=v^*+w_8r^*+\omega_9$. To determine the stability domain associated with $(u^*,v^*,r^*)$, we rewrite the dispersion relation as a cubic polynomial function given as
\begin{align}\label{eq:1.1.11}
P(\lambda(k^2))=\lambda^3 + \boldsymbol {\mu_2}(k^2)\lambda^2 + \boldsymbol {\mu_1}(k^2)\lambda + \boldsymbol {\mu_0}(k^2),
\end{align}
with coefficients
\begin{align*}
\boldsymbol {\mu_2}(k^2)&=(d_1 + d_2 + d_3)k^2 - (J_{11} + J_{22} + J_{33} ) ,\\
\boldsymbol {\mu_1}(k^2)&=J_{11}J_{33} + J_{11}J_{22} + J_{22}J_{33} - J_{32}J_{23} - J_{12}J_{21} \\
& - k^2\big( (d_3 + d_1)J_{22} + (d_2 + d_1)J_{33}  + (d_2 + d_3)J_{11}\big) \\
 &+k^4(d_2d_3 + d_2d_1 + d_1d_3),\\
\boldsymbol {\mu_0}(k^2)&=J_{11} J_{32} J_{23}  - J_{11}J_{22}J_{33} +  J_{12} J_{21} J_{33} \\
& + k^2\big(d_1( J_{22} J_{33} - J_{32} J_{23} ) + d_2 J_{11} J_{33} + d_3 ( J_{22} J_{11} - J_{12} J_{21} )\big)\\
& - k^4\big( d_2d_1J_{33} + d_1d_3J_{22} + d_2d_3J_{11}\big) + k^6d_1d_2d_3.
\end{align*}
According to Routh-Hurwitz criterion for stability, $ \mathbb{R}e(\lambda)<0$ in model \eqref{eq:1dc}  around equilibrium point $(u^*,v^*,r^*)$ (i.e $(u^*,v^*,r^*)$ is stable) if and only if these conditions hold:
\begin{align}\label{eq:1.1.12}
  \boldsymbol {\mu_2}(k^2)>0,\,\boldsymbol {\mu_1}(k^2)>0,\,\boldsymbol {\mu_0}(k^2)>0\quad \text{and}\quad [\boldsymbol {\mu_2}\boldsymbol {\mu_1}-\boldsymbol {\mu_0}](k^2) >0.
\end{align}\\
Where as violating either of the above conditions implies instability (i.e $\mathbb{R}e(\lambda)>0$).
We now require conditions where an homogeneous steady state $(u^*,v^*,r^*)$ will be stable to small perturbation in the absence of diffusion and unstable in the present of diffusion with certain $k$ values. Meaning, we require that around the homogeneous steady state $(u^*,v^*,r^*)$
\begin{align*}
 \mathbb{R}e(\lambda(k^2>0))>0,\, \text{for some}\, k \, \text{and}\, \mathbb{R}e(\lambda(k^2=0))<0,
\end{align*}
where we consider $k$ to be real and positive even though $k$ can be complex. This behavior is called \emph{Diffusion driven instability}. Models that exhibits this behavior in $2$ and $3$ species have been extensively studied in \cite{22,G98, xie12}, were several different patterns was observed.
In order for homogeneous steady state $(u^*,v^*,r^*)$  to be stable (in the absence of diffusion) we need
\begin{align*}
\boldsymbol {\mu_2}(k^2=0)>0,\,\boldsymbol {\mu_1}(k^2)>0,\,\boldsymbol {\mu_0}(k^2=0)>0\quad \text{and}\quad [\boldsymbol {\mu_2}\boldsymbol {\mu_1}-\boldsymbol {\mu_0}](k^2=0) >0,
\end{align*}
whereas with diffusion ($k^2>0$) we look for conditions where we can drive the homogeneous steady state to be unstable, this can be achieved by studying the coefficient of \eqref{eq:1.1.11}. In order to achieve this we reverse at least one of the signs in \eqref{eq:1.1.12}. For this  we have to first study $\boldsymbol {\mu_2}(k^2)$. Irrespective of the value of $k^2$, $\boldsymbol {\mu_2}(k^2)$ will be positive since $J_{11}+J_{22}+J_{33}$ is always less than zero. Therefore we cannot depend on $\boldsymbol {\mu_2}(k^2)$ for diffusion driven instability to occur. Hence for diffusion driven instability to occur in our case, we only depend on the  $2$ conditions which are
\begin{align}\label{eq:1.1.12a}
\boldsymbol {\mu_0}(k^2)\quad \text{and}\quad [\boldsymbol {\mu_2}\boldsymbol {\mu_1}-\boldsymbol {\mu_0}](k^2).
\end{align}
Both functions are cubic functions of $k^2$, which are generally of the form
\begin{align*}
G(k^2)=H_H + k^2D_D + (k^2)^2C_C + (k^2)^3B_B,\, \text{with}\, B_B>0,\,\text{and} \, H_H>0.
\end{align*}
The coefficient of $G(k^2)$ are standard, see \cite{G98}.

To drive either  $\boldsymbol {\mu_0}(k^2)$ or $[\boldsymbol {\mu_2}\boldsymbol {\mu_1}-\boldsymbol {\mu_0}](k^2)$ to negative for some $k$, we basically need to find the minimum $k^2$ referred to as the minimum turing point ($k^2_T$) such that $G(k^2=k^2_T)<0$. This minimum Turing point occurs when
$$ {\partial G / \partial (k^2)}=0,$$
which when solved for $k^2$ we obtain
\begin{align*}
k^2=k^2_T={-C_C + \sqrt{C_C^2 - 3B_BD_D}   \over 3B_B},
\end{align*}
which ensures $k^2$ is real and positive such that $ {\partial^2 G / \partial^2 (k^2)}>0$, by which we require either
\begin{align}\label{eq:1.1.13}
D_D<0\quad \text{or}\quad C_C<0,
\end{align}
which ensures that
\begin{align*}
C_C^2 - 3B_BD_D>0.
\end{align*}
Therefore $G(k^2)<0$, if at $k^2=k^2_T$
\begin{align}\label{eq:1.1.14}
G_{min}(k^2)=2C_C^3 - 9D_DC_CB_B - 2(C_C^2- 3D_DB_B)^{3/2} + 27B_B^2H_H<0.
\end{align}
Hence  \eqref{eq:1.1.13}-\eqref{eq:1.1.14}
 are necessary and sufficient conditions for $(u^*,v^*,r^*)$ to produce diffusion driven instability, which leads to emergence of patterns. Also to first establish stability when $k=0$, $H_H$  in each case has to be positive.

\subsection{Numerical Result:}
Here we demonstrate Turing patterns that form in 1D. The initial condition used is a small perturbation around the positive homogeneous steady state given as 
\begin{align*}
u=u^{*} + \epsilon_1  cos^2(nx)(x > 0)(x < \pi),\\
v=v^{*}  + \epsilon_2 cos^2(nx)(x > 0)(x < \pi),\\
r=r^{*}  + \epsilon_3  cos^2(nx)(x > 0)(x < \pi),
\end{align*}
where $ \epsilon_i=0.005$ $\forall i$. 

\begin{figure}[htb]
{
\includegraphics[width =3in]{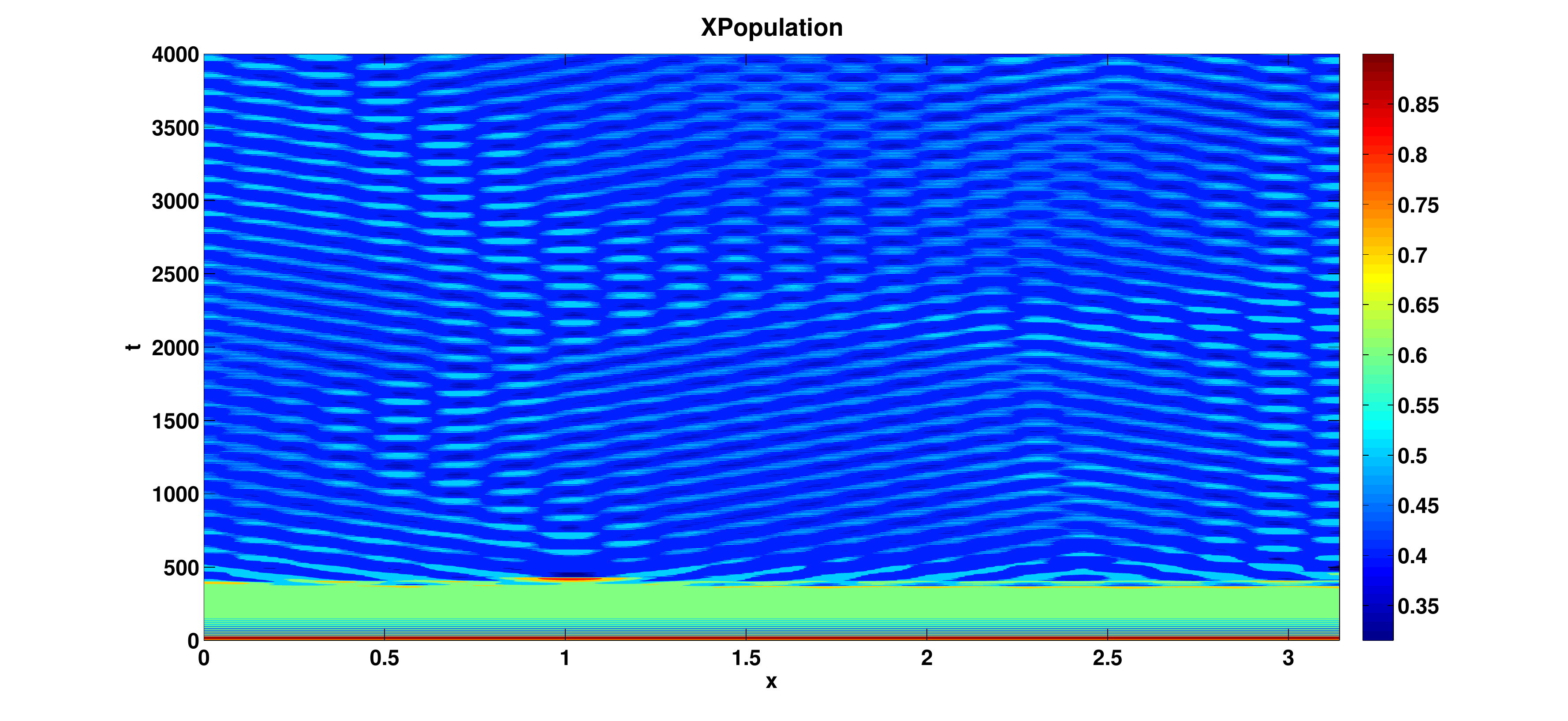}}
{
\includegraphics[width =3in]{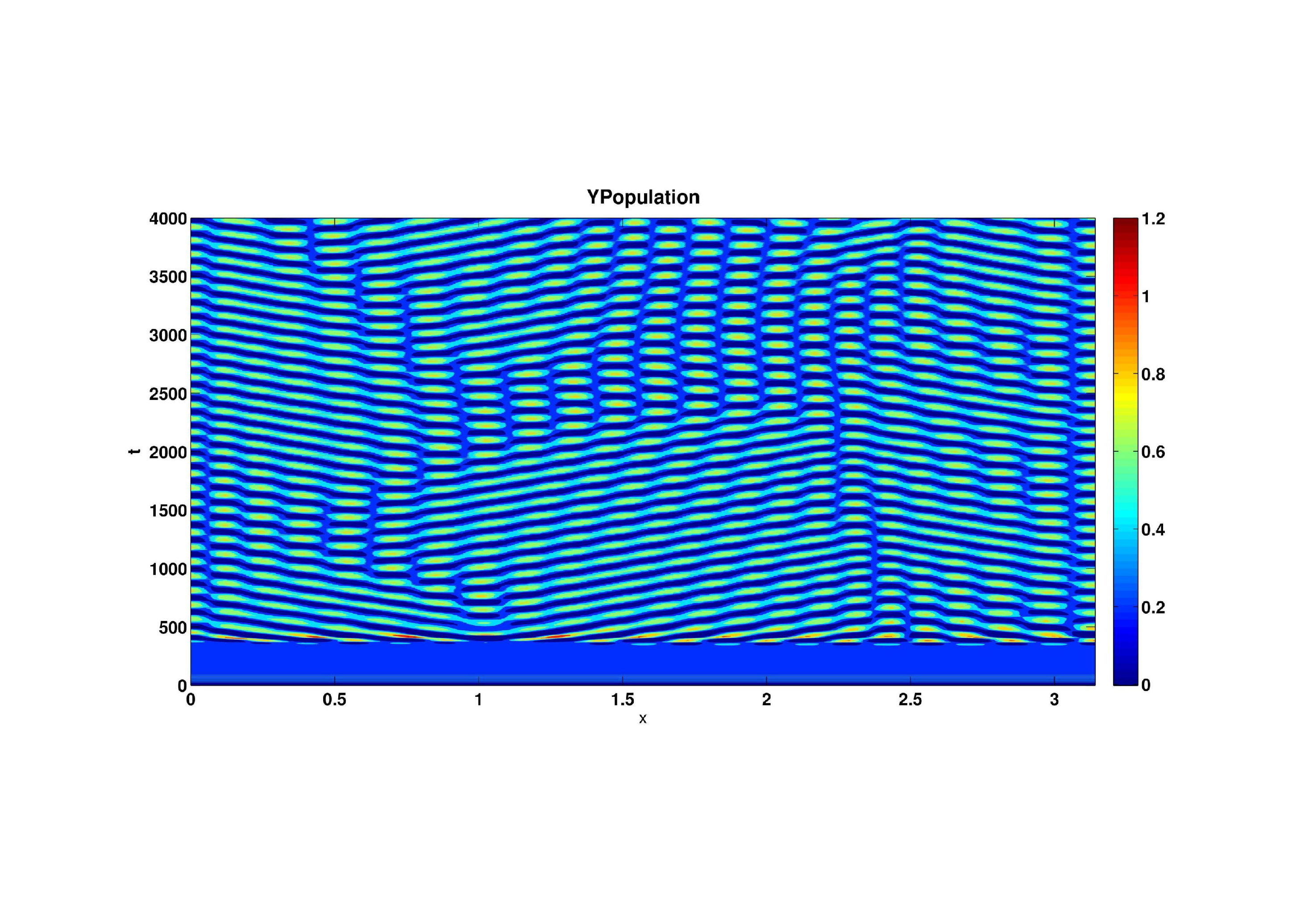}}
{
\includegraphics[width =3in]{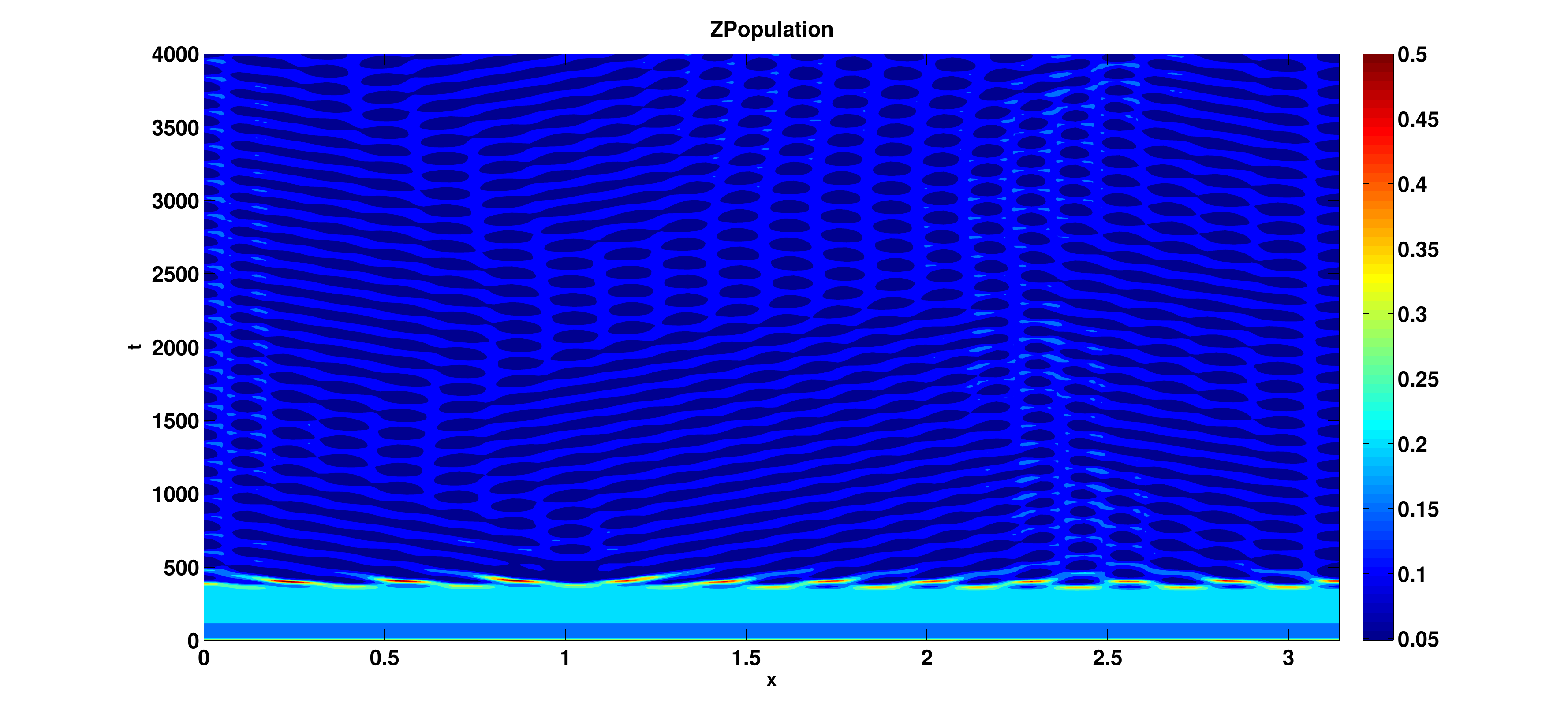}}
\caption{The densities of the three species are shown as contour plots in the x-t plane (1 dimensional in space). The long-time simulation yields Spatio-temporal patterns, that are spatio-temporal.  }
\label{1t}
\end{figure}

\begin{figure}[htb]
{
\includegraphics[width =3in]{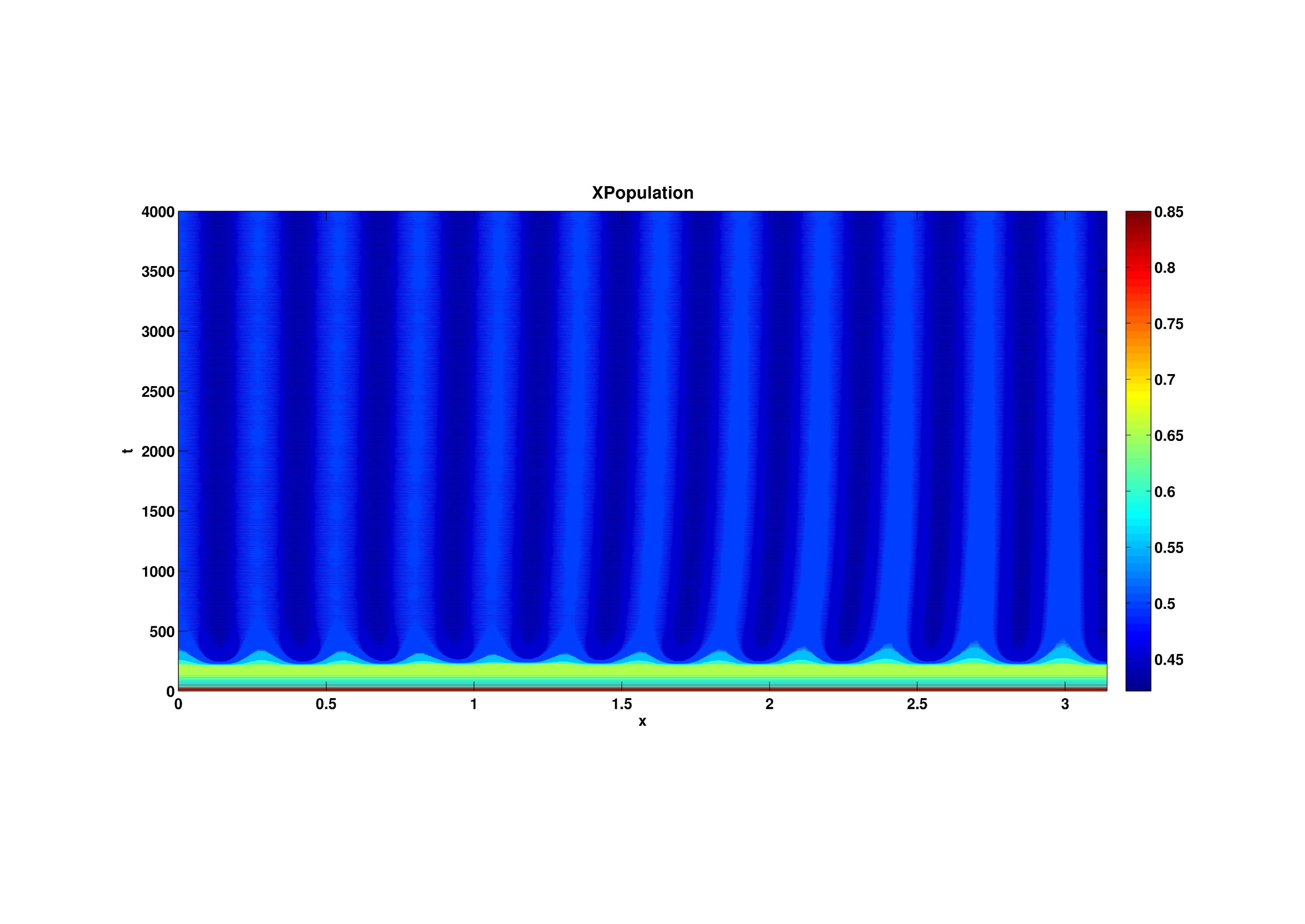}}
{
\includegraphics[width =3in]{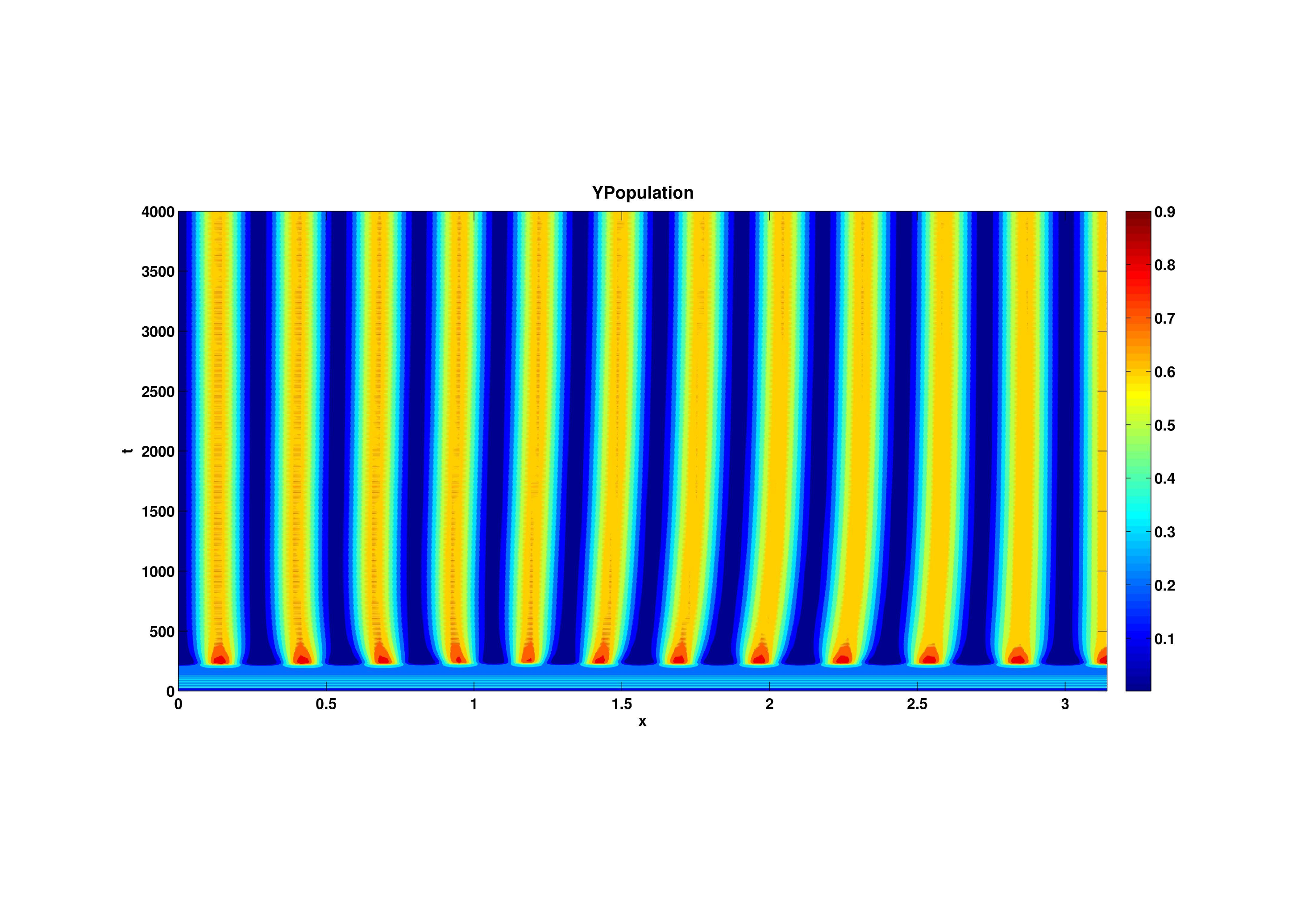}}
{
\includegraphics[width =3in]{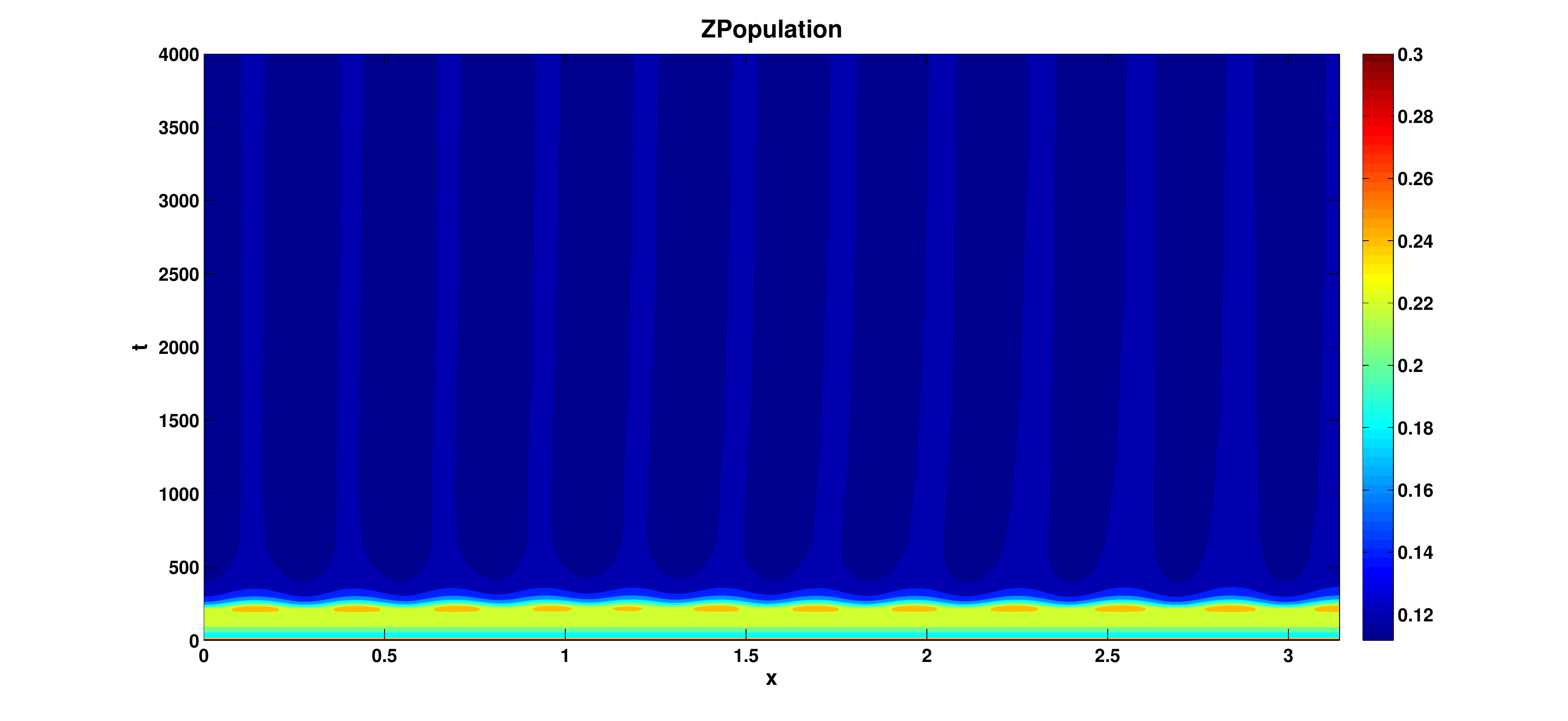}}
\caption{The densities of the three species are shown as contour plots in the x-t plane (1 dimensional in space). The long-time simulation yields stripe patterns that are purely spatial.}
\label{2t}
\end{figure}

\begin{figure}[htb]
{
\includegraphics[width =2in]{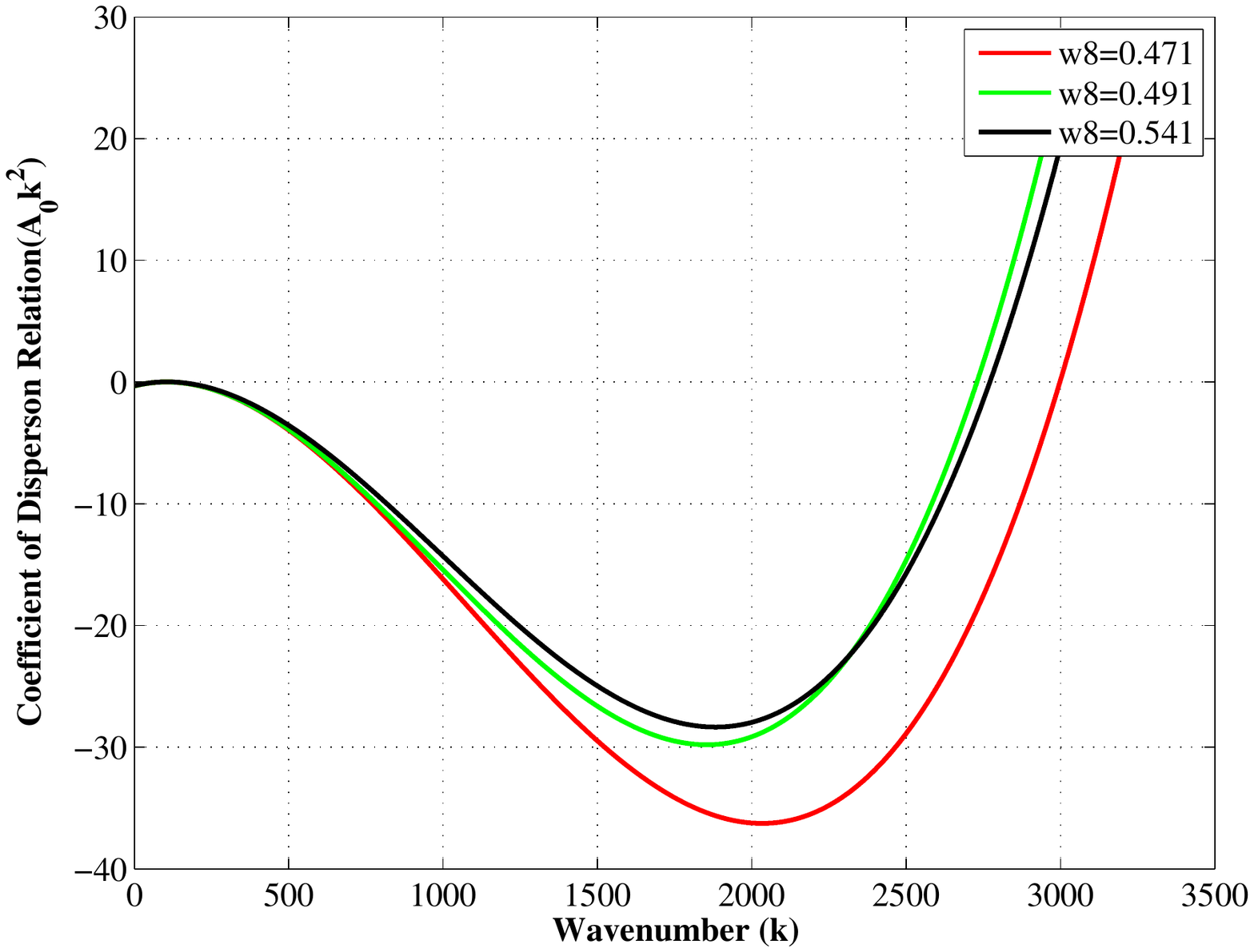}}
\caption{We observe the effect of \emph{predator interference} given by parameter $\omega_{8}$ on the Turing instability. This plot shows that increasing    $\omega_8$, makes fewer modes unstable and thus effects the Turing instability.}
\label{3t}
\end{figure}
\label{app1}

We choose parameters  $ \omega_4=0.210,  \omega_5=1.150,  \omega_6=2.930,  \omega_7=0.540, \omega_8=0.491, \omega_9=0.100, \omega_{10}=0.350, \omega_{11}=0.200, \omega_{12}=0.280, d_1 =0.01, d_2 =0.0001, d_3=0.000001, u^*=0.771530, v^*=0.29, r^*=0.313$, and simulate \eqref{eq:1dc} to obtain spatiotemporal patterns as seen in fig. \ref{1t}. We next choose parameter values  The  parameters are: $  \omega_4=0.210, \omega_5=1.160, \omega_6=2.930, \omega_7=0.540, \omega_8=0.541, \omega_9=0.100, \omega_{10}=0.350, \omega_{11}=0.200, \omega_{12}=0.280, d_1 =0.01, d_2 =0.0001, d_3 =0.001,  u^*=0.771530,  v^*=0.291,  r^*=0.317$ and simulate \eqref{eq:1dc} to obtain spatial patterns as seen in fig. \ref{2t}. In fig \ref{3t} we see that increasing interference causes fewer modes to become unstable and thus effects the Turing instability. 

\section{Spatio-Temporal Chaos}
The goal of this section is to investigate spatio-temporal chaos in the model \eqref{eq:1dc}. Spatio-temporal chaos is usually defined as deterministic dynamics in spatially extended systems that are characterized by an apparent randomness in space and time \cite{Cai01}. There is a large literature on spatio-temporal chaos in PDE, in particular there has been a recent interest on spatially extended systems in ecology exhibiting spatio-temporal chaos \cite{M02}. However, most of these works are on two species models, and there is not much literature in the three-species case. Note, that the appearance of a jagged structure in the species density, as seen in \cite{N13}, which seems to change in time in an irregular way, does not necessarily mean that the dynamics are chaotic. One rigorous definition of chaos means sensitivity to initial conditions. Thus two initial distributions, close together, should yield an exponentially growing difference in the species distribution at later time. In order to confirm this in \eqref{eq:1dc}, we perform a number of tests as in \cite{M02}. We run \eqref{eq:1dc} from a number of different initial conditions, that are the same modulo a small perturbation. We then look at the difference of the two densities, at each time step in both the $L^{\infty}$ and $L^{1}$ norms.\\
The simulations use two different (but close together in $L^{1}(\Omega), L^{2}(\Omega), L^{\infty}(\Omega)$ norms) initial conditions. The first simulation (which we call $r_{unpert}$) is a perturbation of $(u^{*},v^{*},r^{*})$ by $0.0055\cos^{2}(x)$.  The second simulation (which we call $r_{pert}$) is a perturbation of $(u^{*},v^{*},r^{*})$ by $0.0056\cos^{2}(x)$. The densities of the species are calculated up to the time $t=10000.$  At each time step in the simulation we compute
\[
d(t) = ||r_{unpert}(x,t) - r_{pert}(x,t)||_{X} ,
\]
\noindent where $X=L^{1}(\Omega),~L^{2}(\Omega) \ \mbox{ and } \ L^{\infty}(\Omega)$ are used.  Then, $d(t)$ is plotted on a log scale.  In doing so, we observe the exponential growth of the error.  This grows at an approximate rate of $0.0082>0$. Since this is positive then this is an indicator of spatio-temporal chaos. These numerical tests provide experimental evidence for the presence of spatio-temporal chaos in the classical model \eqref{eq:1dc}. Figure \ref{ContourChaos} shows the densities of the populations in the $xt$-plane while Figure \ref{ContourChaosError} gives the error and its logarithm till $t=1000$.
\begin{figure}[htb!]
     \includegraphics[scale=.3]{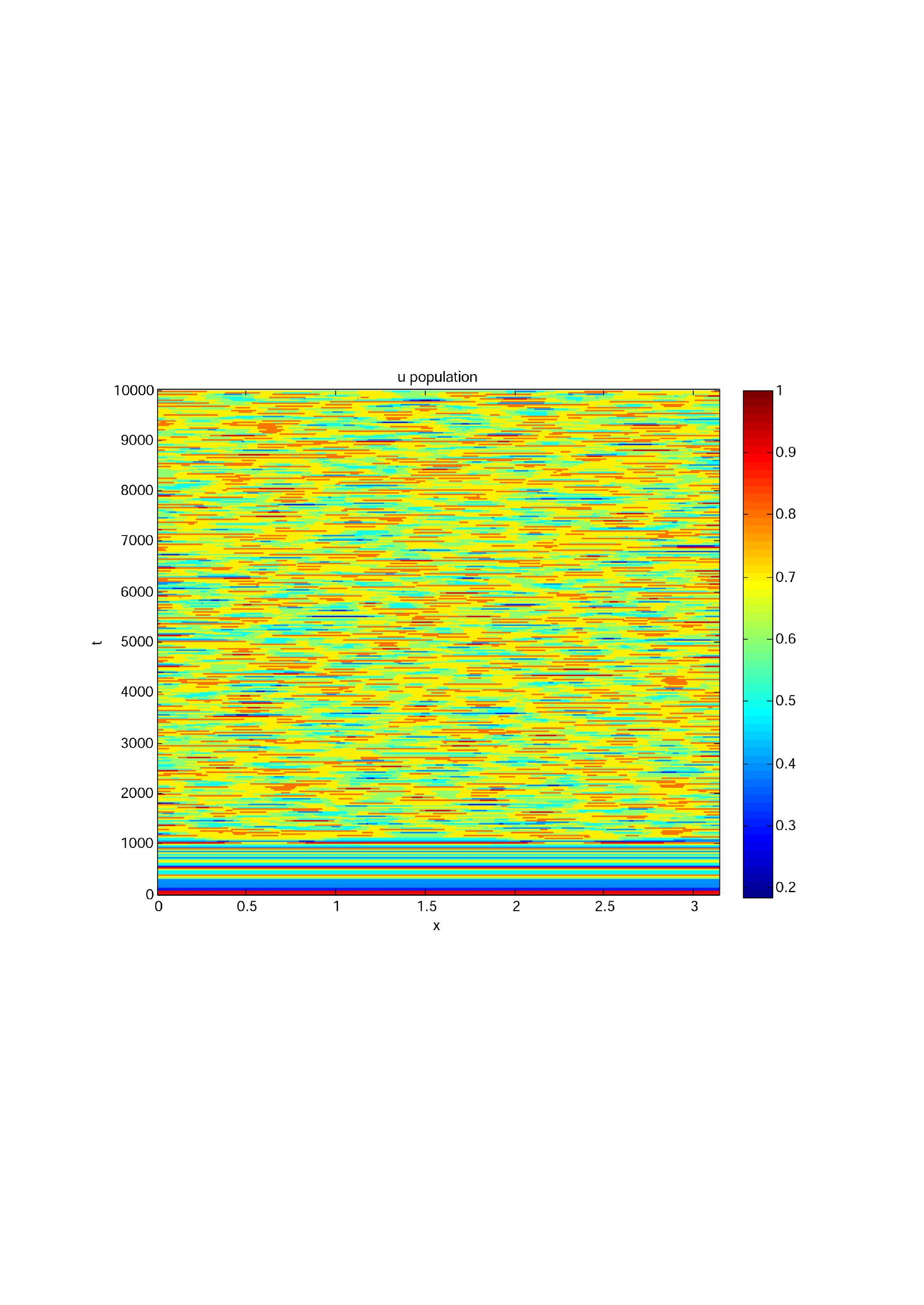}
	\caption{ The densities of the species $u$ is shown as contour plots in the $xt$-plane. The long-time simulation yields spatio-temporal chaotic patterns. $128$ grid points are used with a temporal step size of $.01$. The parameters used are $\omega_4=0.210, \omega_5=1.150, \omega_6=2.930, \omega_7=0.540, \omega_8=0.491, \omega_9=0.100, \omega_{10}=0.350, \omega_{11}=0.200, \omega_{12}=0.280,  d_1 =0.01, d_2 =0.0001, d_3 =0.000001, u^*=0.771530, v^*=0.291,  r^*=0.313$.  }
	\label{ContourChaos}
\end{figure}
\begin{figure}[htb!]
	\begin{center}
             \includegraphics[scale=0.3]{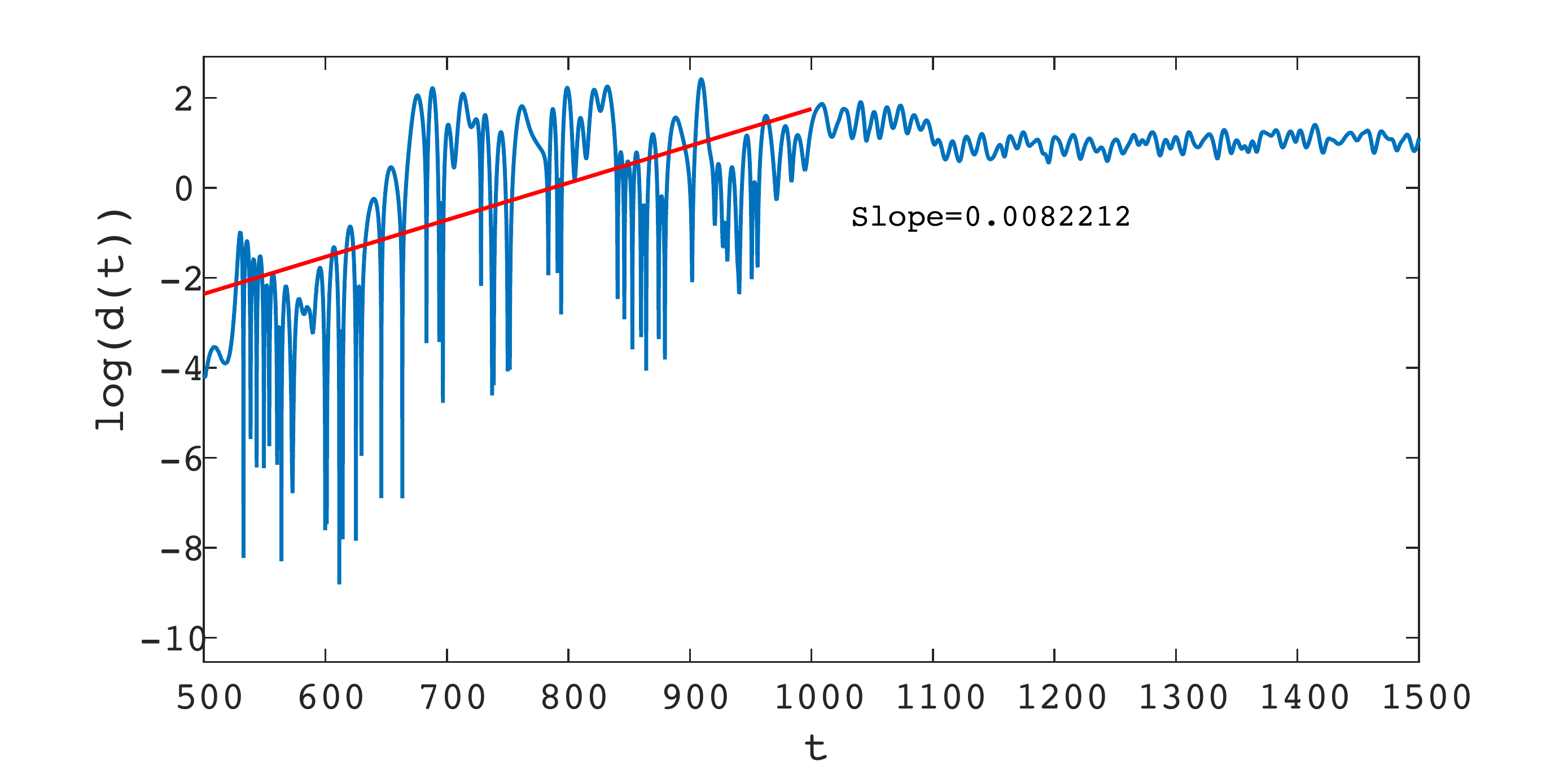}
		\includegraphics[scale=0.3]{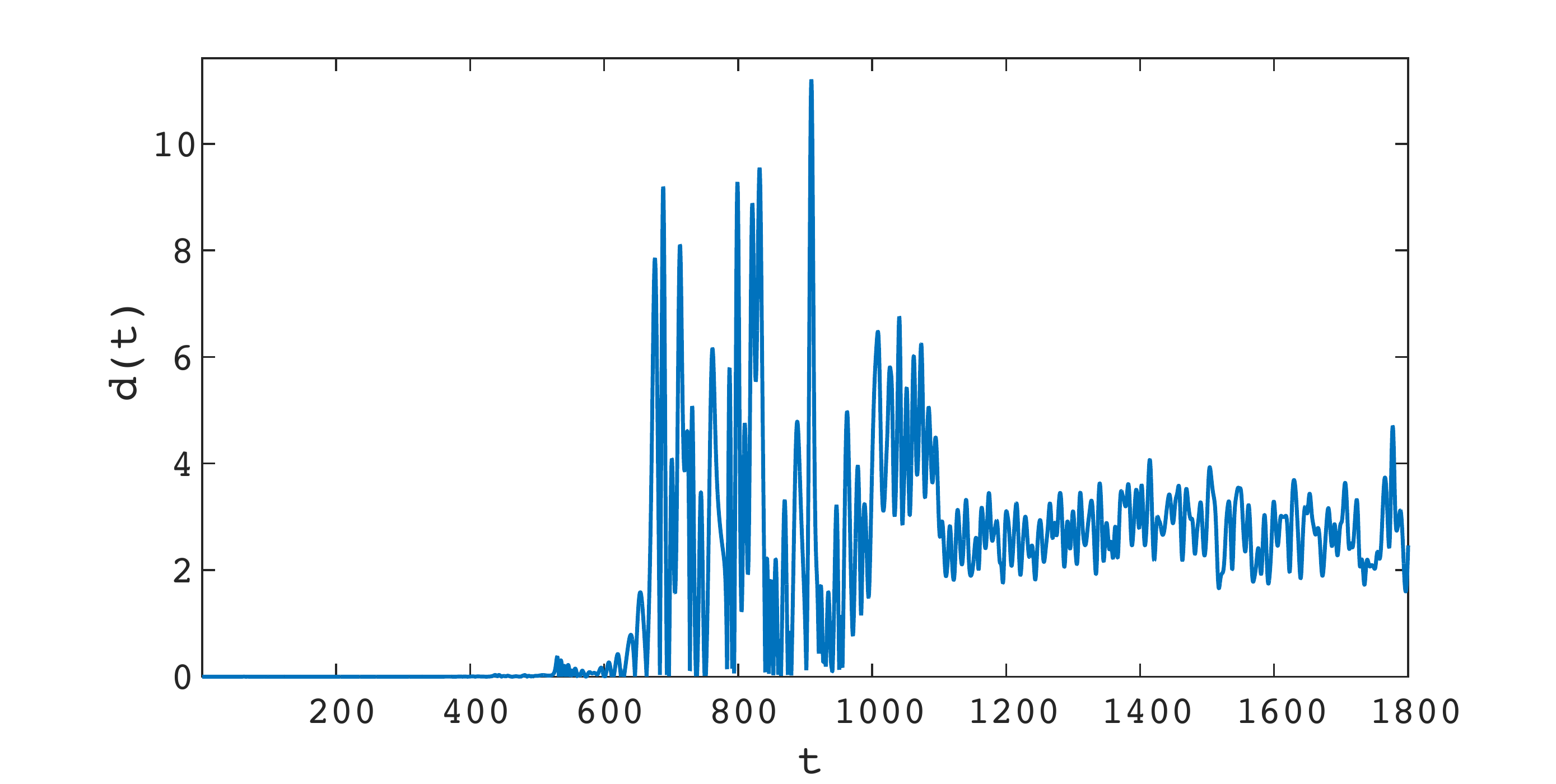}
		\end{center}
	\caption{Here we show the growth of the error $d(t) = ||r_{unpert}(x,t) - r_{pert}(x,t)||_{L^{2}}$, and its logarithm. The error grows at an approximate rate of $0.0082>0$, confirming spatio-temporal chaos.}
	\label{ContourChaosError}
\end{figure}

\section{Effect of Time Delay}
\label{3}
In this section we will try to unravel the effect of delay both on the finite time blow-up dynamics and on the chaotic dynamics. There is a large literature on the effect of delay on both two and three species predator-prey models. 
Upadhyay and Agrawal \cite{18} investigated the effect of mutual interference on the dynamics of delay induced predator prey system, and determined the conditions under which the model becomes globally asymptotically stable around the nonzero equilibria. Recently, Jana et al. \cite{19} have made an attempt to understand the role of top predator interference and gestation delay on the dynamics of a three species food chain model. Interaction between the prey and the middle predator follows Volterra scheme, while that between the top predator and its prey depends on the Beddington-DeAngelis type functional response. Upadhyay et al. \cite{20} studied the three species food chain model with generalist type top predator and obtained that increasing the top predator interference stabilizes the system, while increasing the normalization of the residual reduction in the top predator population destabilizes the system. In our current investigations we choose a constant time delay, in various forms. These are demonstrated next. We perform all our simulations using the standard MATLAB routine DDE23, for delay differential equations \cite{S021}.

\subsection{Effect of Delay on the Chaotic dynamics}
In this section we attempt to numerically investigate the effect of a constant time delay $\tau$ on the chaotic dynamics the system possesses \cite{20}. For this we choose to place the time delay in different parts of the functional response, of the top predator equation only. The first delayed model we consider is the following;
\begin{equation}
\label{eq:1deq}
\frac{dr}{dt}=\omega_{10}r^2-\frac{\omega_{11}r^2(t-\tau)}{v(t-\tau)+\omega_{12}}
\end{equation}
The equations for $u,v$ remain the same.
 For the investigations we have chosen $\tau=1.9$ in figure \ref{1de} and $\tau=.09$ in fig \ref{2de}.
\begin{figure}[htb!]
	\begin{center}
	\includegraphics[scale=0.27]{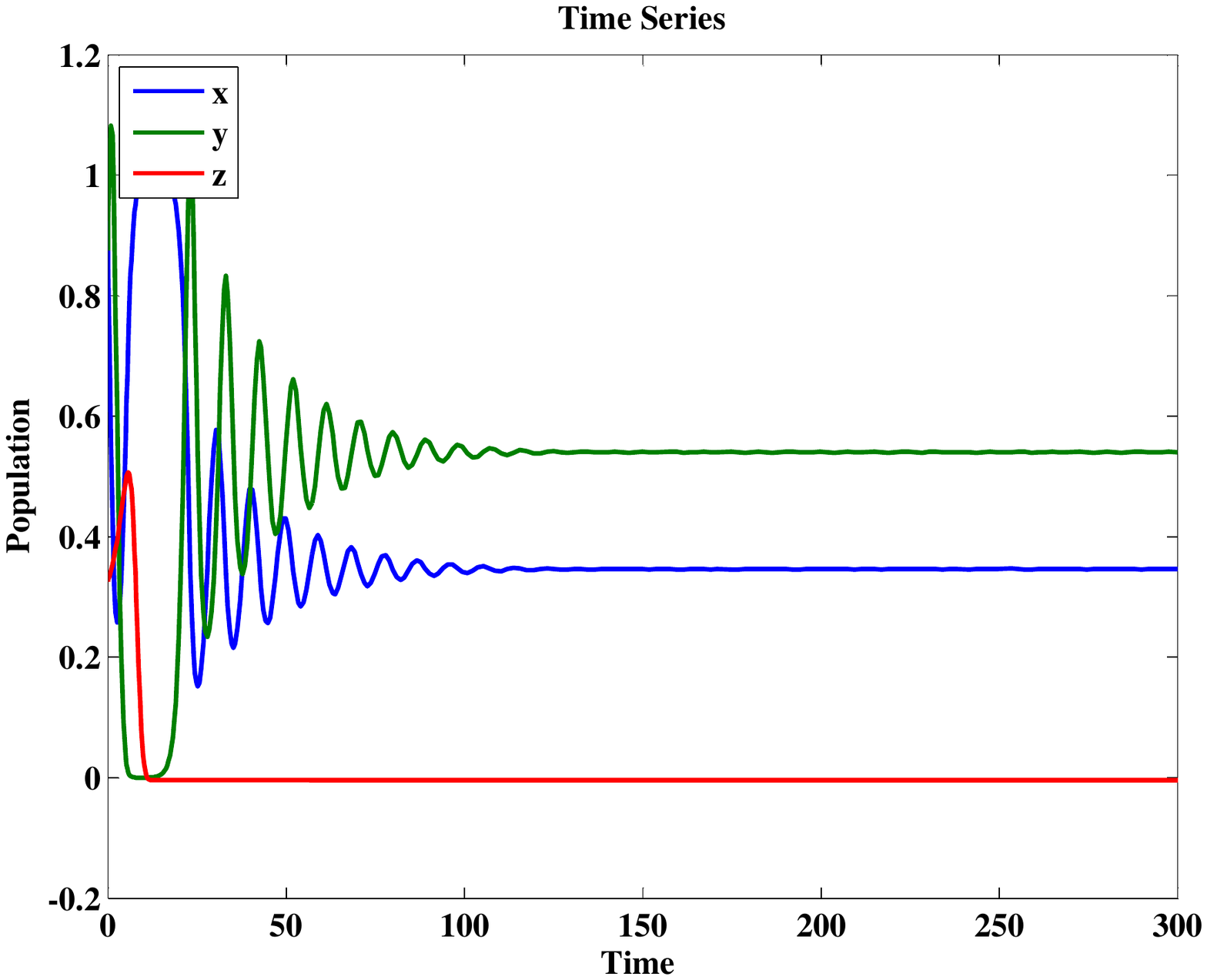}
    \includegraphics[scale=0.27]{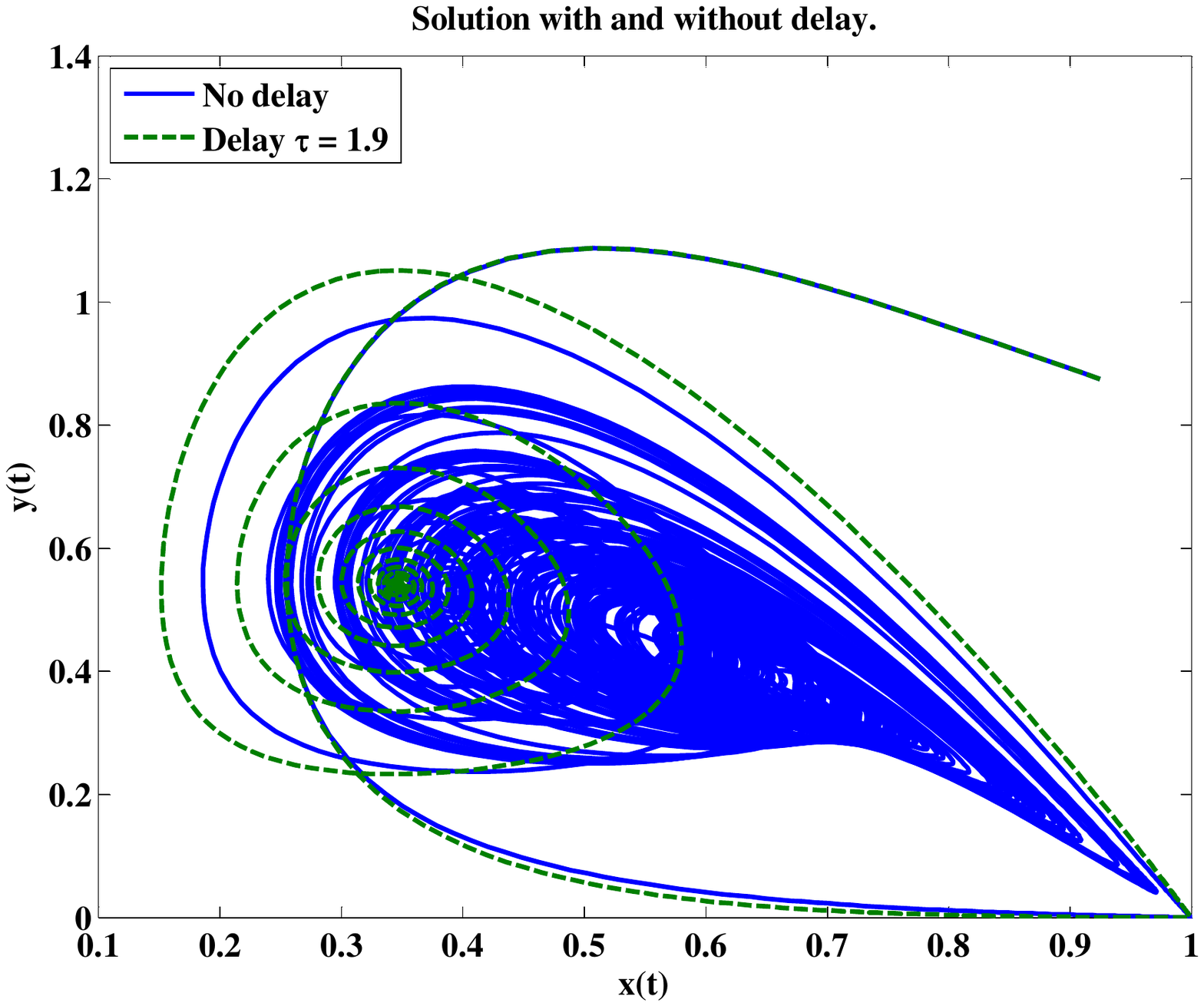}
	\end{center}
\caption{In this scenario we observe that if the time delay is modeled via \eqref{eq:1deq} then a delay of $\tau=1.9$ takes a chaotic state (with no delay) into a stable focus.}
\label{1de}
\end{figure}

\begin{figure}[htp!]
	\begin{center}
	\includegraphics[scale=0.27]{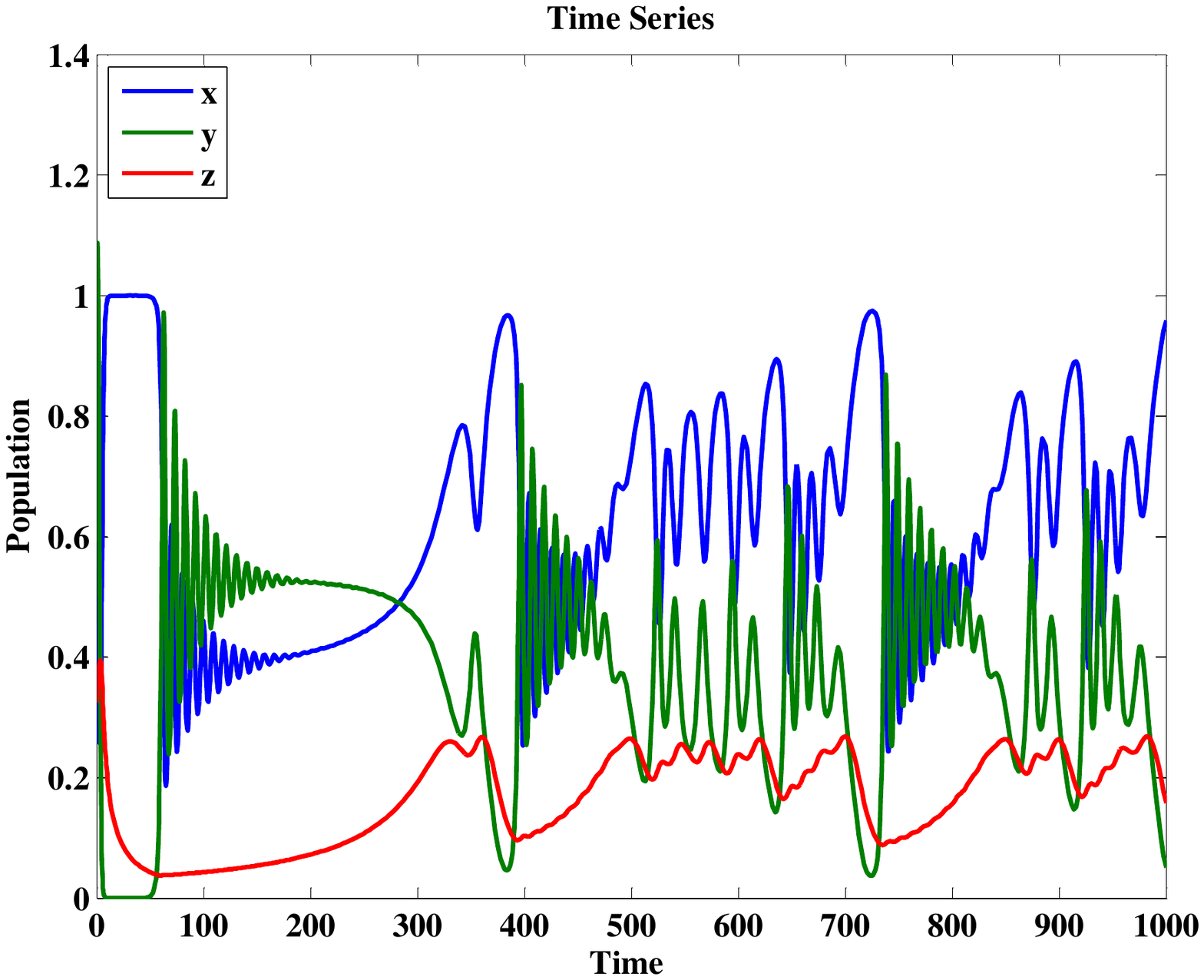}
    \includegraphics[scale=0.27]{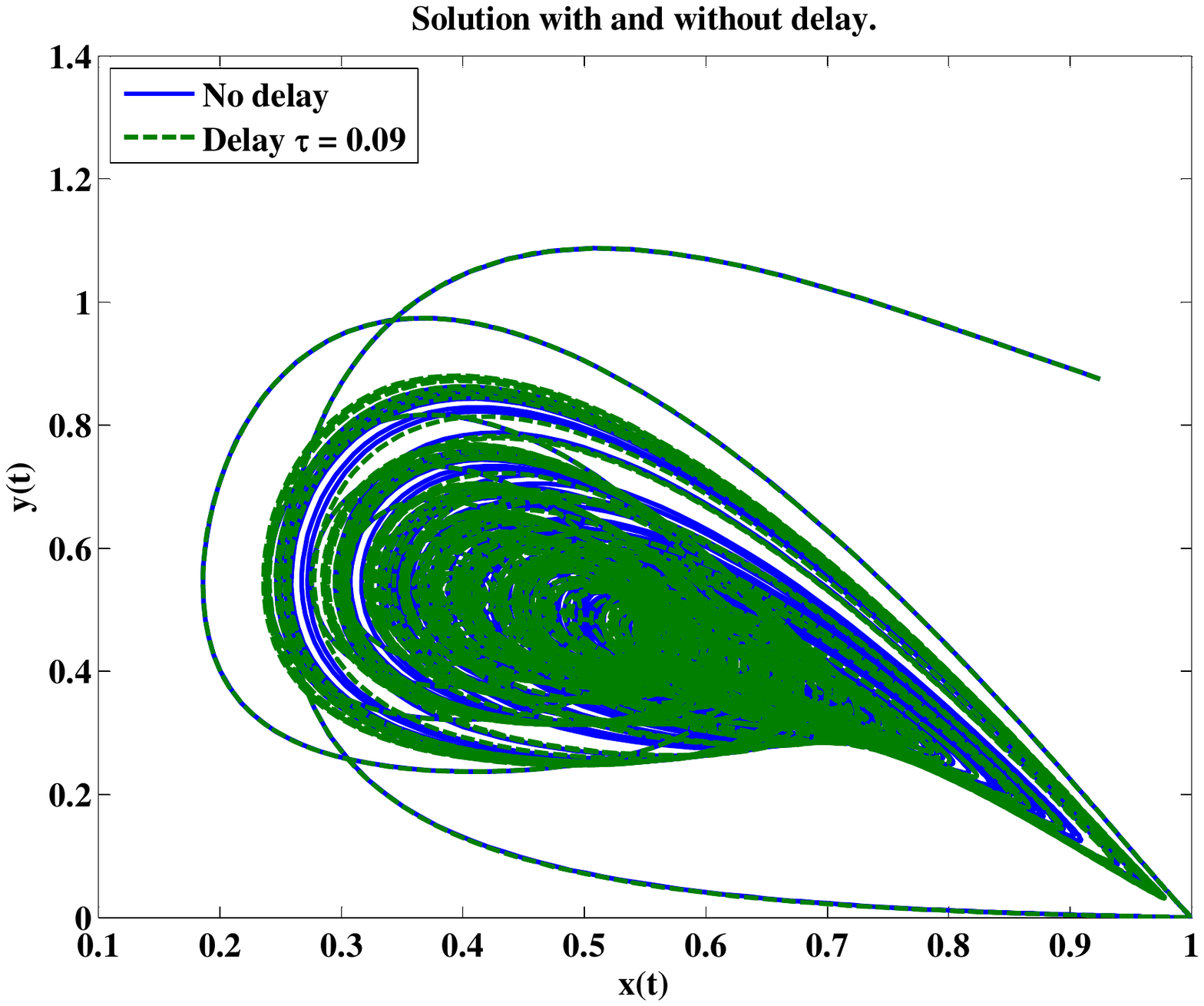}
	\end{center}
\caption{In this scenario we observe that if the time delay is modeled via \eqref{eq:1deq} then a small delay of $\tau=.09$ maintains a similar structure.}
\label{2de}
\end{figure}

Therefore via fig.\ref{1de}, fig.\ref{2de} we see that for a small time delay, the dynamics remain the same, but with an increase in the delay we can observe a radical change in the dynamics.

Next we aim to observe the effect of delay on the blow-up dynamics in the model. Here we incorporated delay in its growth term $\omega_{10}r^2$. This seems plausible due to gestation effect, as well as because this is the term that causes finite time blow-up. 
After the introduction of delay into the model we have;
\begin{equation}
\label{eq:1ade}
\frac{dr}{dt}=\omega_{10}r^2(t-\tau)-\frac{\omega_{11}r^2}{v+\omega_{12}}
\end{equation}

\begin{figure}[htb!]
\begin{center}
	\includegraphics[scale=0.27]{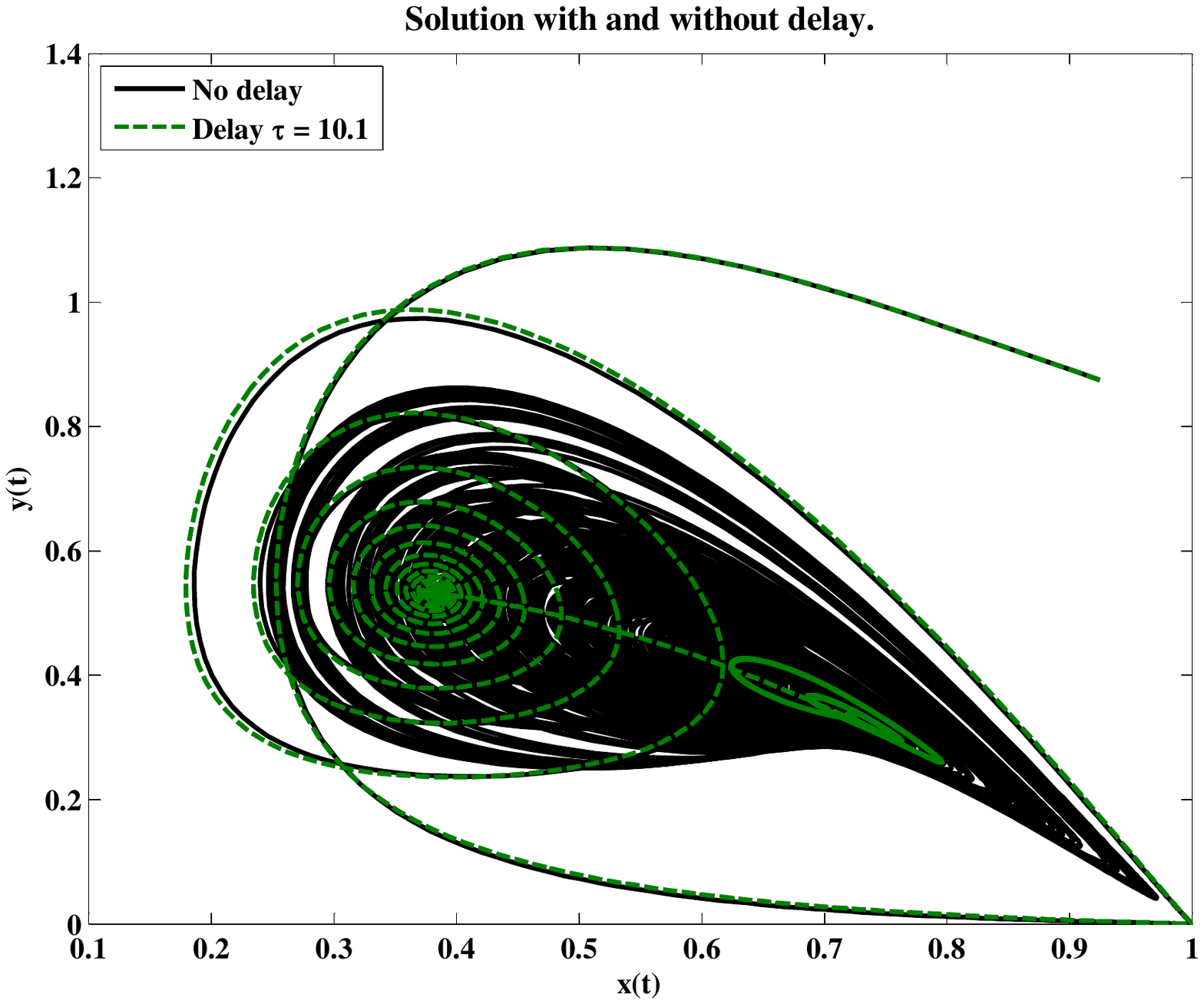}
	\includegraphics[scale=0.27]{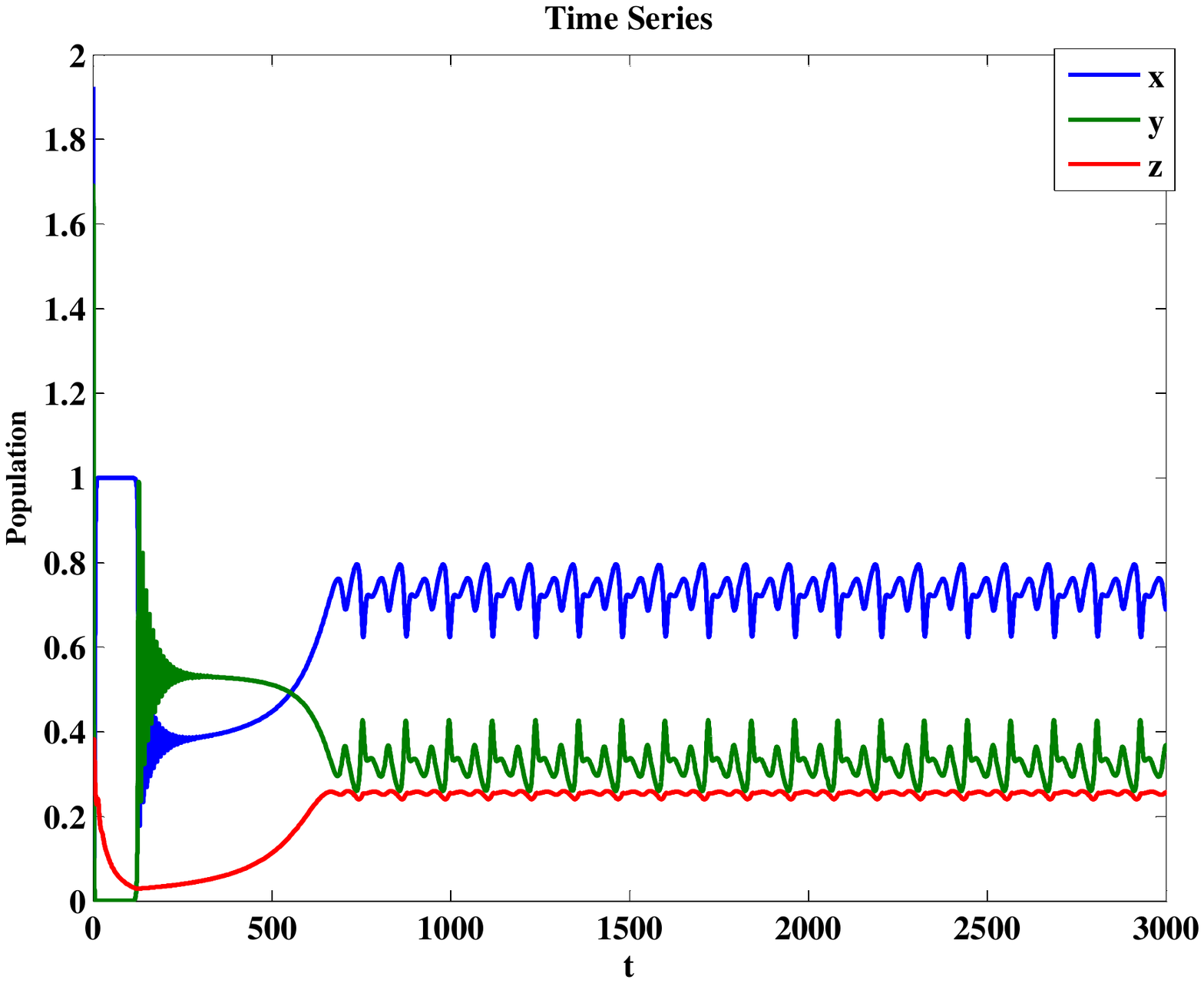}
		\end{center}
\caption{Here we see that a large delay of $\tau=10.1$ in \eqref{eq:1ade} takes a chaotic state (no delay) into a limit cycle.}
\label{fig:3de}
\end{figure}
We have chosen a parameter set for which we have the blow up in the ODE model without delay. Upon introducing the delay ($\tau=10.1$), in \eqref{eq:1ade} we found that the system does not exhibit blow-up.\\
The chosen parameter set for all of our simulations is
$\omega_4=0.210; \omega_5=1.290;\omega_6=2.930;\omega_7=0.540;\omega_8=0.981;\omega_9=0.100;\omega_{10}=0.445;\omega_{11}=0.200;\omega_{12}=0.280.$


%
%

\section{Conclusion}

In the current manuscript we propose an alternate theory that provides a partial answer to the paradox of the generalist predator. Generalist predators are considered poor for biological control purposes, primarily due to mutual interference and their interference in the search of other specialist predators. How then might they be effective in controlling pest densities, as suggested by real field data? We suggest that the interference, might actually be a cause in their population explosion, enabling them in these excessive numbers to control the target pest. From a biological point of view, I  think the crucial point is that by interfering in the search of specialist predators, they \emph{indirectly keep the pest density high enough}, for themselves to excessively harvest the pest, resulting in a sharp growth of their own population, described mathematically by finite time blow-up. 
Thus there seems to be an underlying feedback mechanism, between this indirect interference and their own harvesting. This is a subtle point that warrants further investigation.

 Our work also opens an alternate approach to understanding the population explosion of species such as the cane toad, introduced originally for biological control. Maybe the toad due to its excessive interference with other predators, was able to keep the population of its food source high enough, so that it could feed enough and grow unchecked. This explosive growth was also helped by the fact that the toad being poisonous, was not easily predated upon. 

Finite time blow-up here should be viewed as a mathematical construct, that is a conduit to understanding population explosions. One should not consider it in the literal sense, as a population cannot become infinite in finite time.  However, an equation describing populations, with such emergent behavior, can be a tool to understand excessive increases in population, such as precisely the situation with the cane toad. All in all we hope that the alternate approach we provide, will help reconcile the conflict between theory and data, as concerns the effectiveness of the generalist predator as a biological control.

\end{document}